\documentclass[10pt,conference]{IEEEtran}
\IEEEoverridecommandlockouts
\IEEEpubid{\hspace{-\columnwidth}/978-1-4799-6204-4/14\$31.00 \copyright2014 IEEE}

\usepackage{graphicx}
\usepackage{caption}
\usepackage{subcaption}
\usepackage{pifont}
\usepackage{url}
\usepackage{amsfonts}
\usepackage{times}
\usepackage{threeparttable}
\usepackage{amsmath, amssymb}
\usepackage{pslatex}
\usepackage{enumerate}
\usepackage{algorithm}
\usepackage{clrscode}
\usepackage{arcs}
\usepackage{yhmath}
\usepackage{color}
\usepackage{colortbl}
\newtheorem{theorem}{Proposition}
\newtheorem{lemma}[theorem]{Lemma}
\usepackage{array}
\usepackage{flushend}

\begin{document}


\title{Space Shuffle: A Scalable, Flexible, and High-Bandwidth Data Center Network}
\author{
Ye Yu and  Chen Qian\\
Department of Computer Science,
University of Kentucky \\
ye.yu@uky.edu, qian@cs.uky.edu\\
}
\maketitle

\begin{abstract}
Data center applications require the network to be scalable and  bandwidth-rich.
Current data center network architectures often use rigid topologies to increase network bandwidth. A major limitation is that they can hardly support incremental network growth. Recent studies propose to use random interconnects to provide growth flexibility. However,  routing on a random topology suffers from control and data plane scalability problems, because routing decisions require global information and forwarding state cannot be aggregated. In this paper, we design a novel flexible data center network architecture, Space Shuffle (S2), which applies greedy routing on multiple ring spaces to achieve high-throughput, scalability, and flexibility. The proposed greedy routing protocol of S2 effectively exploits the path diversity of densely connected topologies and enables key-based routing.  Extensive experimental studies show that S2
provides high bisectional bandwidth and throughput, near-optimal routing path lengths, extremely small forwarding state,  fairness among concurrent data flows, and resiliency to network failures.

\end{abstract}


\newtheorem{mydef}{theorem}

\section{Introduction}  
\label{sec:intro}

Data center networks, being an important computing and communication component for cloud services and big data processing, require high inter-server communication bandwidth and scalability \cite{DCN-survey}. Network topology and the corresponding routing protocol are determinate factors of application performance in a data center network. 
Recent work has been investigating new topologies and routing protocols with a goal of improving network performance in the following aspects.

 1) \textbf{High-bandwidth:} Many applications of current data center networks are data-intensive and require
substantial intra-network communication, such as MapReduce \cite{mapreduce}, Hadoop \cite{hadoop}, and Dryad \cite{Dryad}. Data center networks should have densely connected topologies which provide high bisection bandwidth and multiple parallel paths between any pair of servers.  Routing protocols that can effectively exploit the network bandwidth and path diversity are essential.

 2) \textbf{Flexibility:} A data center network may change after its deployment. According to a very recent survey \cite{DigitalRealty}, 93\% US data center operators and 88\% European data center operators will definitely or probably expand their data centers in 2013 or 2014. Therefore a data center network should support \emph{incremental growth} of network size, i.e., adding servers and network bandwidth incrementally to the data center network without destroying the current topology or replacing the current switches.

 3) \textbf{Scalability:} Routing and forwarding in a data center network should rely on small forwarding state of switches and be scalable to large networks. Forwarding table scalability is highly desired in large enterprise and data center networks, because they use expensive and power-hungry memory to achieve increasingly fast line speed \cite{BUFFALO} \cite{SWDC} \cite{ROME}. If forwarding state is small and does not increase with the network size, we can use relatively inexpensive switches to construct large data centers and do not need switch memory upgrade when the network grows.


Unfortunately, existing data center network architectures  \cite{fattree} \cite{vl2} \cite{Portland} \cite{DCell} \cite{Symbiotic} \cite{SWDC} \cite{Jellyfish} focus on one or two of the above properties and pay little attention to the others. For example, the widely used multi-rooted tree topologies \cite{fattree} \cite{Portland} provide rich bandwidth and efficient routing, but their ``firm'' structures cannot deal with incremental growth of network size. The recently proposed Jellyfish network \cite{Jellyfish} uses random interconnect to support incremental growth and near-optimal bandwidth \cite{Godfrey14}. However, Jellyfish has to use inefficient $k$-shortest path routing whose forwarding state is big and cannot be aggregated. CamCube \cite{Symbiotic}  and Small World Data Centers (SWDC) \cite{SWDC} propose to use greedy routing for forwarding state scalability and efficient key-value services. Their greedy routing protocols do not produce shortest paths and can hardly be extended to perform multi-path routing that can fully utilize network bandwidth.

\begin{table*}[t]
\centering
\small 
\caption{
Desired properties of data center network architectures. 
 $N$: \# switches, $M$: \# links.
 Question mark means such property is not discussed in
the paper. 
}
\vspace{-1ex}
\label{table:properties}
\begin{tabular}{c||ccccc}
\hline
 & FatTree  \cite{fattree} & CamCube \cite{Symbiotic} & SWDC \cite{SWDC} & Jellyfish \cite{Jellyfish} & S2  \\
\hline
Network bandwidth & Benchmark & No Comparison & $>$ Camcube & $>$ FatTree and SWDC & $\approx$ Jellyfish \\
Multi-path routing & \checkmark & ? & ? & \checkmark & \checkmark \\
Incremental growth & \ding{53} & ? & ? & \checkmark & \checkmark\\
Forwarding state per switch & $O(\log N)$ & constant & constant & $O(kN\log N)$ & constant \\
Key-based routing & \ding{53} & \checkmark & \checkmark & \ding{53} & \checkmark\\
Switch heterogeneity & \ding{53} & \ding{53} & \ding{53} & \checkmark & \checkmark\\
 \hline
\end{tabular}
\vspace{-5ex}
\end{table*}

Designing a data center network that satisfies all three requirements seems to be challenging. Flexibility requires irregularity of network topologies, whereas high-throughput routing protocols on irregular topologies, such as $k$-shortest path, are hard to scale.
In this paper, we present a new data center network architecture, named Space Shuffle (S2), including \emph{a scalable greedy routing protocol that achieves high-throughput and near-optimal path lengths on flexible and bandwidth-rich networks built by random interconnection.}

S2 networks are constructed by interconnecting an arbitrary number of commodity ToR switches.
Switches maintain coordinates in multiple \emph{virtual spaces}.
%
We also design a novel greedy routing protocol called \emph{greediest routing} that guarantees to find multiple paths to any destination on an S2 topology.
Unlike existing greedy routing protocols \cite{GDV,MDT-ToN}, which use only one single space, greediest routing makes decisions by considering switches coordinates in multiple spaces.
The routing path lengths are close to shortest path lengths.
In addition, coordinates in multiple spaces enable efficient and high-throughput multi-path routing of S2.
S2 also effectively supports key-based routing, which has demonstrated to fit many current data center applications using key-value stores \cite{Symbiotic}.
%

%

Table \ref{table:properties} compares S2 and four other recent data center networks qualitatively in seven desired properties, namely high bandwidth, multi-path routing, flexibility for incremental growth, small forwarding state, 
 key-based routing, and support of switch  heterogeneity.
S2 achieves almost all desired properties while every other design has a few disadvantages.

We  use extensive simulation results to demonstrate S2's performance in different dimensions, including routing path length, bisection bandwidth, throughput of single-path and multi-path routing, fairness among flows, forwarding table size, and resiliency to network failures.
Compared to two recently proposed data center networks \cite{SWDC} \cite{Jellyfish}, S2 provides significant advantages in some performance dimensions and is equally good in other dimensions.

The rest of this paper is organized as follows.
We present related work in Section \ref{sec:related}.
We describe the S2 topology and its construction in Section \ref{sec:topo}.
In Section \ref{sec:routing}, we present the routing protocols and design considerations.
We evaluate the performance of S2 in Section \ref{sec:evaluation}.
We discuss a number of practical issues in Section \ref{sec:discussion} and finally conclude this work in Section \ref{sec:conclusion}.

\section{Related Work}
\label{sec:related}

Recent studies have proposed a number of new network topologies to improve data center performance such as bisection bandwidth, flexibility, and failure resilience.
Al-Fares \emph{et.al.} \cite{fattree} propose a multi-rooted tree structure called FatTree
that provides multiple equal paths between any pair of servers and can be built with commodity switches.
VL2 \cite{vl2} is a data center network that uses flat addresses and provide layer-2 semantics.
Its topology is a Clos network which is also a multi-rooted tree \cite{clos}.
Some data center network designs use direct server-to-server connection in regular topologies to achieve high bisection bandwidth, including DCell \cite{DCell}, BCube \cite{BCube}, CamCube \cite{Symbiotic}, and Small-World data centers \cite{SWDC}.
However, none of these designs have considered the requirement of incremental growth of data centers.

A number of solutions have been proposed to provide network flexibility and support incremental growth.
Scafida \cite{Scafida} uses randomness to build an asymmetric data center network that can be scaled in
smaller increments.
In LEGUP \cite{LEGUP}, free ports are preserved for future expansion of Clos networks.
%
%
REWRITE \cite{REWRITE} is a framework that uses local search to find a network topology that maximizes bisection bandwidth whiling minimizing latency with a give cost budget.
None of these three \cite{Scafida} \cite{LEGUP} \cite{REWRITE} have explicit routing design to utilize the network bandwidth of the irregular topologies.
Jellyfish \cite{Jellyfish} is a recently proposed data center network architecture that applies random connections to allow arbitrary network size and incremental growth.
Jellyfish can be built with any number of switches and servers and can incorporate additional devices by slightly changing the current network.
Using $k$-shortest path routing, Jellyfish achieves higher network throughput compared to FatTree \cite{fattree} and supports more servers than a FatTree using the same number of switches.
However, to support $k$-shortest path routing on a random interconnect, forwarding state in Jellyfish switches is big and cannot be aggregated.
Using the MPLS implementation of $k$-shortest path as suggested in \cite{Jellyfish}, the expected number of forwarding entries per switch is proportional to $kN\log N$, where $N$ is the number of switches in the network.
In addition, $k$-shortest path algorithm is extremely time consuming. Its complexity is $O(kN(M+NlogN))$ for a single source ($M$ is the number of links) \cite{routingmesh}. This may result in slow convergence under network dynamics.
Hence, Jellyfish may suffer from both \emph{data plane} and \emph{control plane scalability} problems.
PAST \cite{PAST} provides another multi-path solution for Jellyfish, but the throughput of Jellyfish may be degraded.
A very recent study \cite{Godfrey14} discusses the near-optimal-throughput topology design for both homogeneous and heterogeneous networks.
It does not provide routing protocols which can achieve the throughput in practice.

As a scalable solution, greedy routing has been applied to enterprise and data center networks  \cite{Symbiotic} \cite{SWDC} \cite{ROME}.
CamCube \cite{Symbiotic} employs greedy routing on a 3D torus topology.
It provides an API for applications to implement their own routing protocols to satisfy specific requirements, called symbiotic routing.
The network topologies of Small-World data centers (SWDCs) are built with directly connected servers in three types: ring, 2D Torus, and 3D Hex Torus.
ROME \cite{ROME} is a network architecture to allow greedy routing on arbitrary network topologies and provide layer-2 semantics. 
For all three network architectures \cite{Symbiotic} \cite{SWDC} \cite{ROME}, multi-path routing is not explicitly provided.

SWDC, Jellyfish, and S2 all employ randomness to build physical topologies. However, they demonstrate substantially different performance because of their different logical organizations and routing protocols. SWDC applies scalable greedy routing on regularly assigned coordinates in a single space and supports key-based routing. Jellyfish provides higher throughput using $k$-shortest path routing, but it sacrifices forwarding table scalability.
S2 gets the best of both worlds: it uses greedy routing on randomly assigned coordinates in multiple spaces to achieve both high-throughput routing and small forwarding state.
%
%

\section{Space Shuffle Data Center Topology}
\label{sec:topo}
The Space Shuffle (S2) topology is a
interconnect of commodity top-of-rack (ToR) switches.
In S2, all switches play a equal role and execute a same protocol.
We assume there is no server multi-homing, i.e., a server only connects with one switch.

\subsection{Virtual coordinates and spaces}

Each switch $s$ is assigned a set of \emph{virtual coordinates} represented by a $L$-dimensional vector $\langle x_{1}, x_{2}, ..., x_{L} \rangle$, where each element $x_{i}$ is a randomly generated real number  $0\leq x_{i}<1$.
%
There are $L$ virtual ring spaces.
In the $i$-th space, a switch is \emph{virtually} placed on a ring based on the value of its $i$-th coordinate $x_i$.
Coordinates in each space are circular, and 0 and 1 are superposed.
Coordinates are distinct in a single space.
In each space, a switch is physically connected with the two adjacent switches on its left and right sides.
Two physically connected switches are called neighbors.
For a network built with $w$-port switches\footnote{We now assume homogenous switches. We will discuss switch heterogeneity in Section \ref{sec:heterogeneity}.}, it is required that $2L < w$.
Each switch has at most $2L$ ports to connect other switches, called inter-switch ports.
The rest ports can be used to connect servers.
A neighbor of a switch $s$ may happen to be adjacent to $s$ in multiple spaces. In such a case, $s$ needs less than $2L$ ports to connect adjacent switches in all $L$ spaces.
Switches with free inter-switch ports can then be connected randomly.

Figure \ref{fig:TopoExample} shows a S2 network  with 9 switches and 18 hosts in two spaces.
As shown in  Figure \ref{fig:connect}, each switch is connected with two hosts and four other switches.
Figure \ref{fig:coord} shows coordinates of each switch in the two spaces.
Figures \ref{fig:space1} and \ref{fig:space2} are the two virtual spaces, where coordinate 0 is at top and coordinates increase clockwisely.
As an example, switch $B$ is connected to switches $A$, $C$, $F$, and $G$, because $A$ and $C$ are adjacent to $B$ in space 1 and $F$ and $G$ are adjacent to $B$ in space 2.
$A$ only uses three ports to connects adjacent switches $I$, $B$, and $H$, because it is adjacent to $I$ in both two spaces.
$A$ and $E$ are connected as they both have free inter-switch ports.


\begin{figure}[t]
        \centering
		\begin{subfigure}[b]{0.4\linewidth}
                \includegraphics[width=\textwidth]{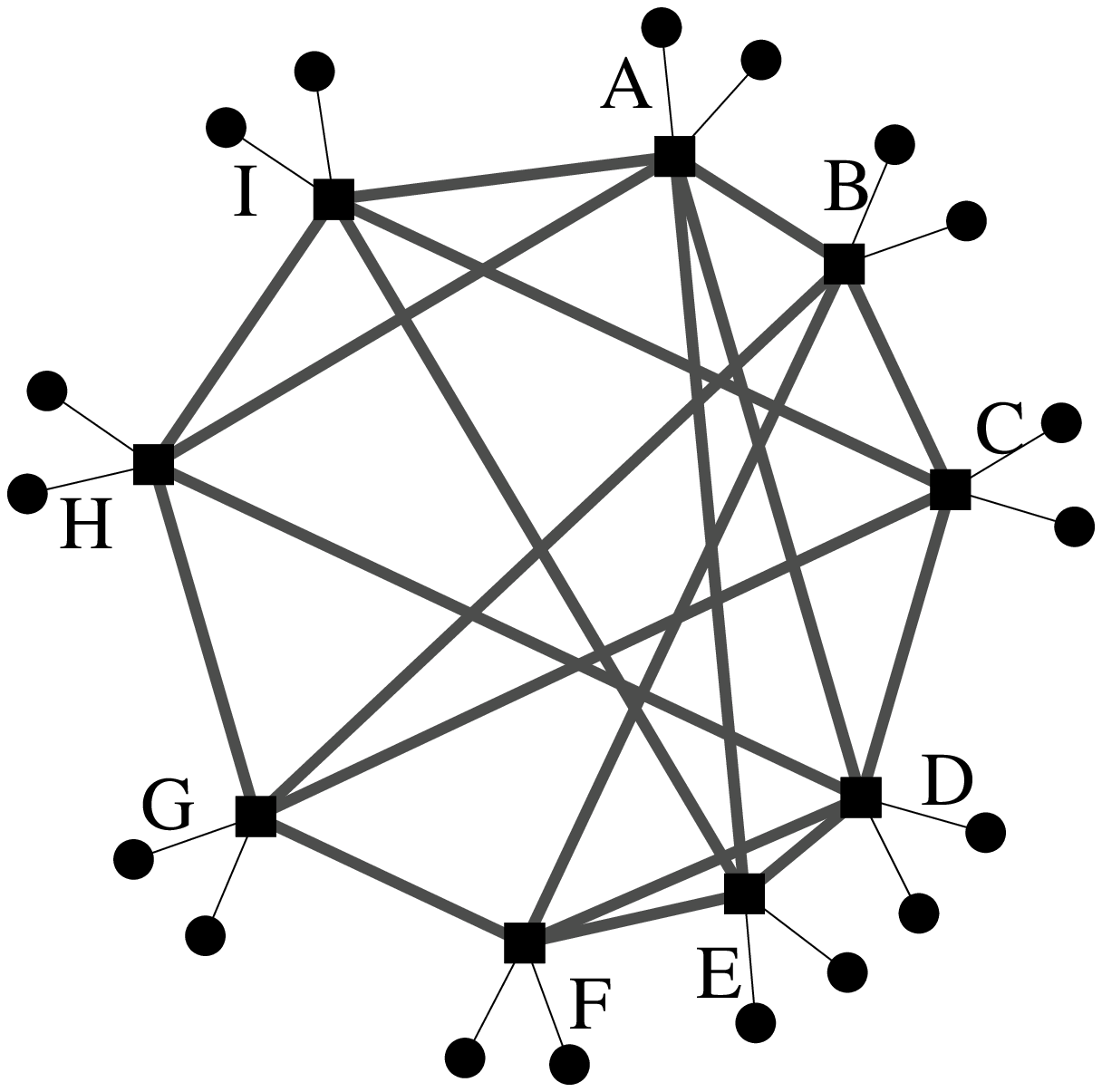}
                \caption{Space Shuffle topology}
                \label{fig:connect}
        \end{subfigure}
        \begin{subfigure}[b]{0.4\linewidth}
            \centering
			\small
            \begin{tabular}{ccc}
            $ID$ & $x_{1}$ & $x_{2}$ \\ \hline
            $A$ & 0.05 &  0.17 \\
            $B$ & 0.13 &  0.62 \\
            $C$ & 0.23 &  0.91 \\
            $D$ & 0.36& 0.42 \\
            $E$ & 0.42 &  0.53 \\
            $F$ & 0.51 &  0.58 \\
            $G$ & 0.63 &  0.73 \\
            $H$ & 0.78 &  0.26 \\
            $I$ & 0.91 &  0.97 \\ \hline
            \end{tabular}
            \caption{Switch coordinates}
            \label{fig:coord}
        \end{subfigure}
		\\
        \begin{subfigure}[b]{0.4\linewidth}
                \includegraphics[width=\linewidth]{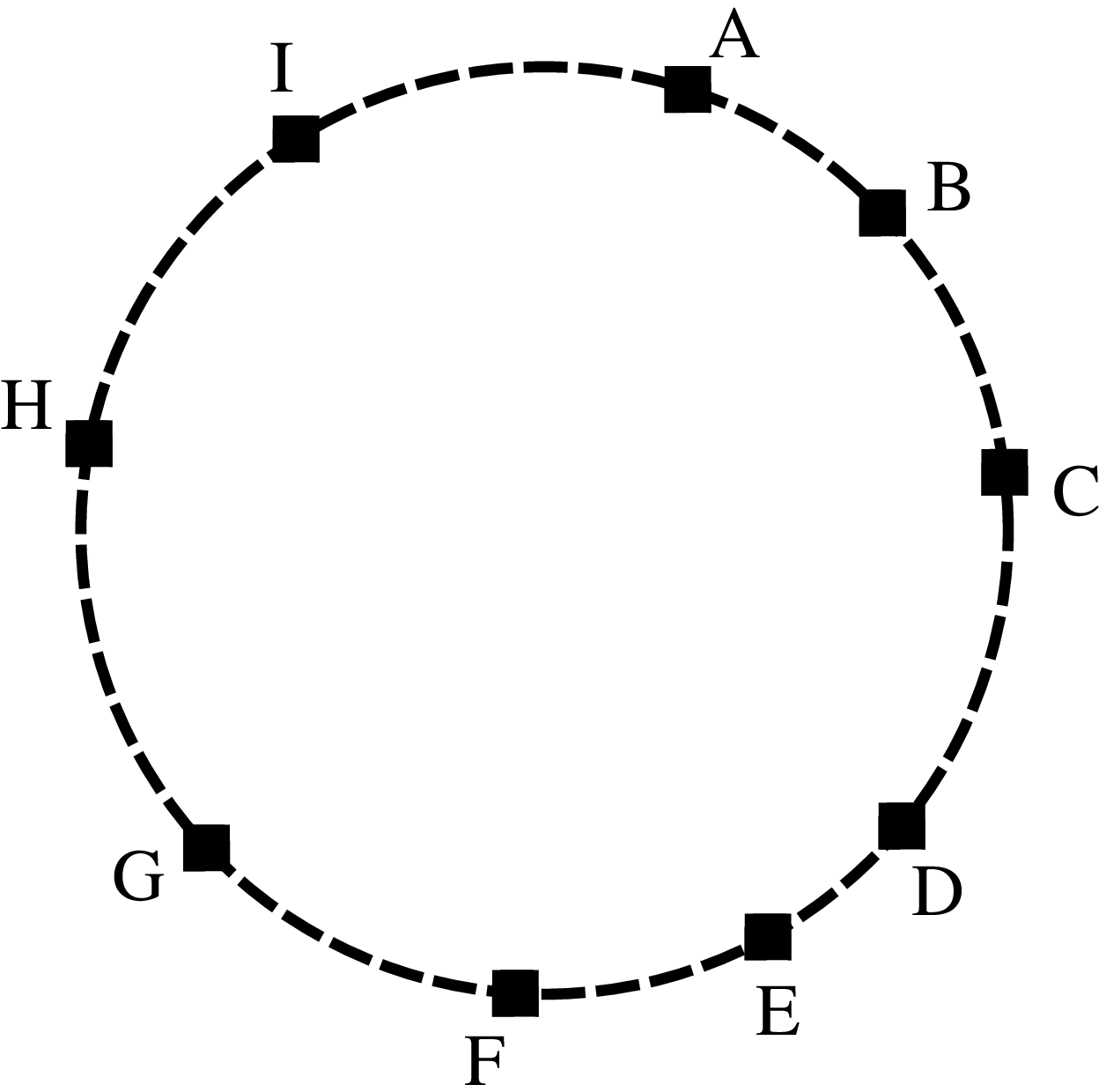}
                \vspace{-3ex}
                \caption{Space 1}
                \label{fig:space1}
        \end{subfigure}%
        \begin{subfigure}[b]{0.4\linewidth}
                \includegraphics[width=\linewidth]{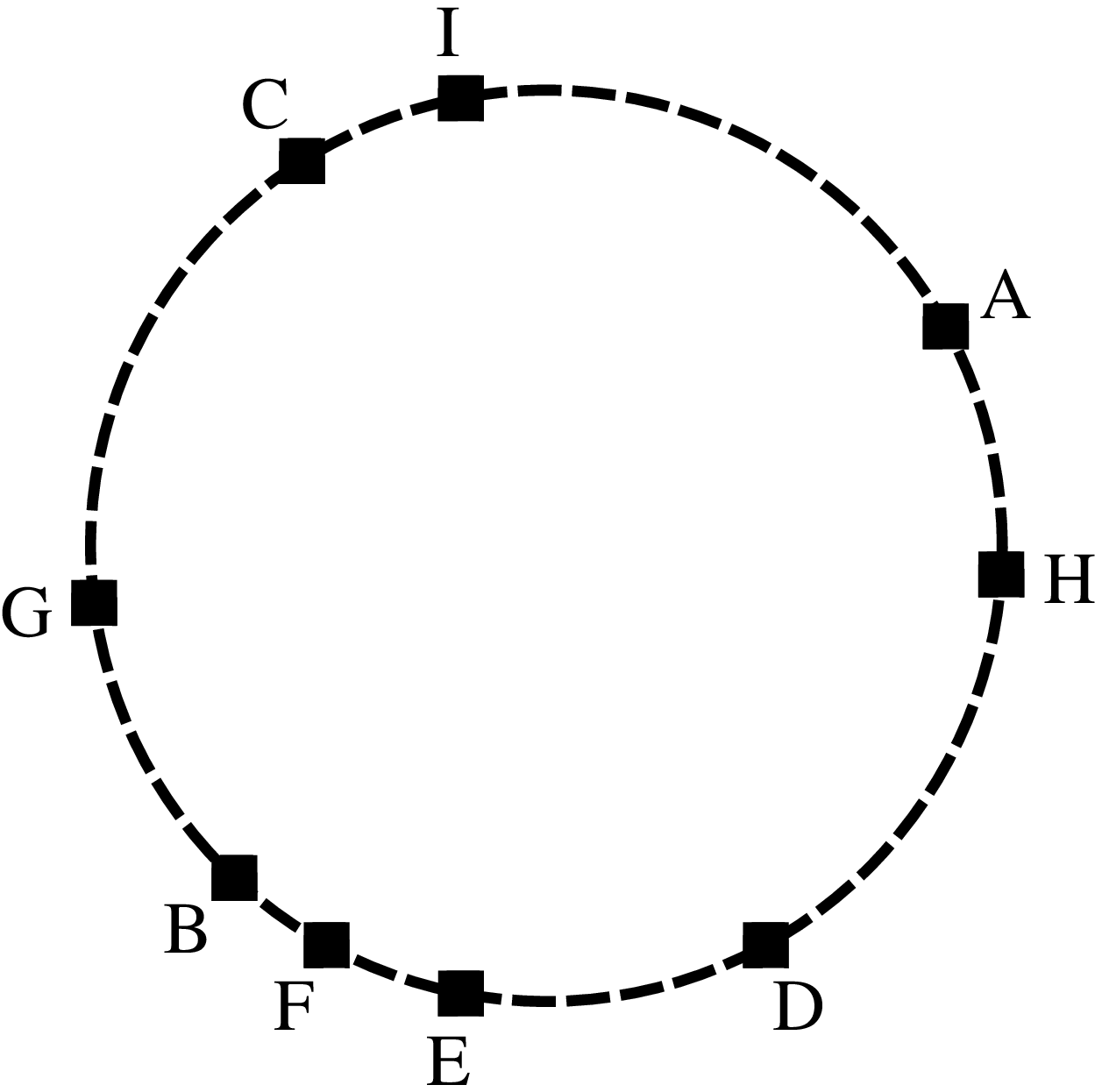}
                \vspace{-3ex}
                \caption{Space 2}
                \label{fig:space2}
        \end{subfigure}
        \vspace{-1ex}
        \caption{Example S2 network with 9 switches and 18 servers in 2 spaces. Squares are switches and circles are servers.}\label{fig:TopoExample}
        \vspace{-5ex}
\end{figure}

\subsection{Topology construction}

As a flexible data center network, S2 can be constructed by either deploy-as-a-whole or incremental deployment.

For the deploy-as-a-whole construction of a network with $N$ switches and $H$ servers, 
each switch is assigned  $\lfloor \frac{H}{N} \rfloor$ or $\lfloor \frac{H}{N} \rfloor + 1$
servers. The number of spaces $L$ is then set to $\lfloor \frac{1}{2} (w - \lceil \frac{H}{N} \rceil)\rfloor$.
Switch positions are randomly assigned in each space.
For each space, cables are placed to connect every pair of adjacent switches.
If there are still more than one switches with free ports, we randomly select switch pairs and connect each pair.
%
%
%
We will discuss more cabling issues in Section \ref{sec:wire}.



S2 can easily support any expansion of the data center network using the incremental deployment algorithm.
Suppose we decide to expand the data center network by $m$ servers.
A switch can connect $w-2L$ servers, and we can determine the number of new switches is $\lceil m/(w-2L) \rceil$.
For each new switch $s$, we assign it a set of random coordinates.
We find $s$'s two adjacent nodes $u$ and $v$ in each space, which is currently connected.
Then, the operator removes the cable between $u$ and $v$ and let $s$ connect to both of them.
New switches and servers can be added serially by iterative execution of this procedure.


Similar to Jellyfish  \cite{Jellyfish}, S2 can be constructed with any number of servers and switches.
For incremental network expansion, only a few cables need to be removed and a few new cables are placed. Hence there is very little network update cost.

At this point, coordinate generation is purely random.
We will discuss the impact of coordinate randomness to the proposed routing protocol and introduce a method to guarantee that any two coordinates are different in Section \ref{sec:balance}.

\subsection{Similar to random regular graphs}
We wonder whether S2 topologies are close to random regular graphs (RRGs), which, as discussed in \cite{Jellyfish} \cite{wang2014similarity} and \cite{Godfrey14}, provide near-optimal bisection bandwidth and lower average shortest path length compared to other existing data center topologies built with identical equipments.
By definition, an $r$-regular graph is a graph where all vertices have an identical degree $r$.
RRGs with degree $r$ are sampled uniformly from the space of all $r$-regular graphs.

Since constructing an RRG is a very complex problem, Jellyfish \cite{Jellyfish} uses the ``sufficiently uniform random graphs'' that empirically have the desired properties of RRGs.
Therefore, we compare S2 with Jellyfish in the average shortest path length.
Table \ref{table:SPcompare} shows the empirical results of shortest path lengths between servers of S2 and Jellyfish.
We show the average, 10\% percentile, and 90\% percentile values for all pairs of servers on 10 different topologies of S2 or Jellyfish.
A network has $N$ switches, each of which has 12 inter-switch ports.
We find that the shortest path lengths of S2 are very close to those of Jellyfish, and they have identical 10\% and 90\% percentile values.
We also find that the  switch-to-switch path lengths of both S2 and Jellyfish follow logarithmic distribution $\log N$, consistent to the property of RRGs \cite{RRGSP}.
As discussed by \cite{Jellyfish}, networks with lower shortest path lengths provide higher bandwidth.
We demonstrate that S2 has almost same shortest path lengths to those of sufficiently uniform random graphs used by Jellyfish.
We will further demonstrate its bisection bandwidth in Section \ref{sec:evaluation}.

Essentially, SWDC, Jellyfish, and S2 use similar random physical interconnects to approximate RRGs\footnote{We also notice a recent work using RRGs for P2P streaming \cite{RRGP2P}, whose routing protocol cannot be used in data center networks.}.
However, their logical organizations and routing protocols are substantially different, which result in different network performance such as throughput and forwarding table size.
\begin{table}[t]
\centering
\small
\caption{Shortest path lengths: S2 vs. Jellyfish}
\label{table:SPcompare}
\vspace{-1ex}
\begin{tabular}{c||c|c|c||c|c|c}
\hline
 & \multicolumn{3}{c||} {\text{SpaceShuffle} }& \multicolumn{3}{c} {\text{JellyFish}}\\
 N  & average & 10\% & 90\% &average & 10\% & 90\% \\
\hline
 100 & 3.80111 & 3 & 4 & 3.80396 & 3 & 4 \\
 200 & 4.00241 & 3 & 5 & 4.00500 & 3 & 5 \\
 400 & 4.29735 & 4 & 5 & 4.29644 & 4 & 5 \\
 800 & 4.57358 & 4 & 5 & 4.57306 & 4 & 5 \\
 1200 & 4.69733 & 4 & 5 & 4.69670 & 4 & 5 \\
 \hline
\end{tabular}
\vspace{-4ex}
\end{table}
\section{Routing Protocols}
\label{sec:routing}
A desired routing protocol in data center networks should have several important features that satisfy application requirements.
First, a routing protocol should guarantee to find a loop-free path to delivery a packet from any source to any destination, i.e.,  \emph{delivery guarantee} and \emph{loop-freedom}.
Second, the routing and forwarding should be scalable to a large size of servers and switches. Third, it should utilize the bandwidth and exploit path diversity of the network topology.

A straightforward way is to use shortest path based routing such as OSPF on S2. However, shortest path routing has a few potential scalability problems.
First, in the data plane, each switch needs to maintain a forwarding table whose size is proportional to the network size.
 The cost of storing the forwarding table in fast memory such as TCAM and SRAM can be high \cite{SWDC}. As the increasing line speeds require the use of faster, expensive, and power-consuming memory, there is a strong motivation to design routing protocol that only uses a small size of memory and does not require memory upgrades when the network size increases \cite{BUFFALO}.
Second, running link-state protocols introduces non-trivial bandwidth cost to the control plane.




\subsection{Greediest Routing}

Since the coordinates of a switch can be considered geographic locations in $L$ different spaces, we design a new greedy geographic routing protocol for S2, called \emph{greediest routing}. 

\textbf{Routable address:} The routable address of a server $h$, namely  $\vec{X}$, is the virtual coordinates of the switch connected to $h$ (also called $h$'s access switch).
Since most current applications uses IP addresses to identify destinations, an address resolution method is needed to obtain the S2 routable address of a packet, as ARP, a central directory, or a DHT \cite{SEATTLE, ROME}. The address resolution function can be deployed on end switches for in-network traffic and on gateway switches for incoming traffic.
In a packet, the destination server $h$  is identified by a tuple $\langle \vec{X}, ID_h \rangle$, where $\vec{X}$ is $h$'s  routable address (virtual coordinates of the access switch) and $ID_h$ is $h$'s identifier such as its MAC or IP address.
The packet is first delivered to the switch $s$ that has the virtual coordinates $\vec{X}$, and then $s$  forwards the packet to $h$ based on $ID_h$.

\textbf{MCD: }We use the \emph{circular distance} to define the distance between two coordinates in a same space.
The circular distance for two coordinates $x$ and $y$ ($0 \leq x,y < 1$) is
\setlength{\abovedisplayskip}{0.1ex}
\setlength{\belowdisplayskip}{0.1ex}
\[CD(x,y) = \min \{|x-y|,1-|x-y|\}\].
In addition, we introduce the \emph{minimum circular distance} (MCD) for routing design.
For two switches $A$ and $B$ with virtual coordinates $\vec{X} = \langle x_1, x_2, ..., x_L \rangle$ and $\vec{Y}= \langle y_{1}, y_{2}, ..., y_{L} \rangle$ respectively, 
the MCD of $A$ and $B$, $MCD(\vec{X}, \vec{Y})$, is the minimum circular distance measured in the $L$ spaces.
Formally,
\[MCD(\vec{X}, \vec{Y}) = \min_{1 \leq i \leq L} CD(x_{i},y_{i}). \]

\textbf{Forwarding decision:}
The greediest routing protocol works as follows.
 When a switch $s$ receives a packet whose destination is $\langle \vec{X}_{t}, ID \rangle$, it first checks whether $\vec{X}_{t}$  is its own coordinates.
 If so, $s$ forwards the packet to the server whose identifier is $ID$. Otherwise, $s$ selects a neighbor $v$ such that $v$ minimizes $MCD(\vec{X}_v,\vec{X}_t)$ to the destination, among all neighbors.
The pseudocode of \textsc{greedist routing on switch} $s$ is presented by Algorithm 1 in the appendix.

\begin{table}[t]
\centering
\small
\caption{MCDs to $C$ from $H$ and its neighbors in Figure \ref{fig:TopoExample}}
\vspace{-1ex}
\label{table:greediest_example}
\begin{tabular}{c||ccc}
\hline
 & Cir dist in Space 1 & Cir dist in Space 2 & Min cir dist \\
\hline
$H$ & 0.45 & 0.35 & 0.35 \\
$A$ & 0.18 & 0.26 & 0.18 \\
$D$ & 0.13 & 0.49 & 0.13 \\
$G$ & 0.40 & 0.18 & 0.18 \\
$I$ & 0.32 & 0.06 & 0.06 \\
 \hline
\end{tabular}
\vspace{-5ex}
\end{table}
For example, in the network shown in Figure \ref{fig:TopoExample},
switch $H$ receives a packet
whose destination host is connected to switch $C$, hence the destination coordinates are $\vec{X}_C$. $H$ has four neighbors $A$, $D$, $I$, and $G$. After computing the MCD from each neighbor to the destination $C$ as listed in Table \ref{table:greediest_example}, $H$ concludes that $I$ has the shortest minimal circular distance to $C$ and then forwards the packet to $I$.

We name our protocol as ``greediest routing'' because it selects a neighbor that has a smallest MCD  to the destination among all neighbors in all spaces.  Existing greedy routing protocols only try to minimize distance to the destination in a single space (Euclidean, or in other kinds).

Greediest routing on S2 topologies provides delivery guarantee and loop-freedom. To prove it, we first introduce two lemmas.
\begin{lemma}
\label{lemma1}
In a space and given a coordinate $x$, if a switch $s$ is not the switch that has the shortest circular distance to $x$ in the space, then $s$ must have an adjacent switch $s'$ such that $CD(x, x_{s'})< CD(x, x_s)$.
\end{lemma}

\begin{lemma}
\label{lemma2}
Suppose switch $s$ receives a packet whose destination switch is $t$ and the coordinates are  $\vec{X}_{t}$, $s \neq t$. Let $v$ be the switch that has the smallest MCD to $\vec{X}_{t}$ among all neighbors of $s$. Then $MCD(\vec{X}_v, \vec{X}_{t}) < MCD(\vec{X}_s, \vec{X}_{t})$.
\end{lemma}

Lemma \ref{lemma2} states that if switch $s$ is not the destination switch, it must find a neighbor $v$ whose MCD is smaller than $s$'s to the destination. Similar to other greedy routing protocols, when we have such ``progressive and distance-reducing'' property, we can establish the proof for delivery guarantee and loop-freedom.
\begin{theorem}
\label{thm:delivery_greediest}
Greediest routing  finds a loop-free path of a finite number of hops to a given destination on an S2 topology.
\end{theorem}

The proofs of the above lemmas and proposition are presented in the appendix.

Like other greedy routing protocols \cite{SWDC, ROME}, greediest routing in S2 is highly scalable and easy to implement. Each switch only needs a small routing table that stores the coordinates of all neighbors. The forwarding decision can be made by a fixed, small number of numerical distance computation and comparisons. More importantly, the routing table size only depends on the number of ports and does not increase when the network grows. In the control plane, decisions are made locally without link-state broadcast in the network wide.

\subsubsection{Reduce routing path length}
An obvious downside of greedy routing is that it does not guarantee shortest routing path.
Non-optimal routing paths incur longer server-to-server latency.
More importantly, flows routed by longer paths will be transmitted on more links, and thus consumes more network bandwidth \cite{Jellyfish}.
To resolve this problem, we allow each switch in S2 stores the coordinates of 2-hop neighbors.
To forward a packet, a switch first determines the  switch $v$ that has the shortest MCD to the destination, among all 1-hop and 2-hop neighbors.
If $v$ is an 1-hop neighbor, the packet is forwarded to $v$.
Otherwise, the packet is forwarded to an one hop neighbor connected to $v$.
Delivery guarantee and loop-freedom still holds.
According to our empirical results, considering 2-hop neighbors can significantly reduce routing path lengths.

As an example, in a 250 10-port switch network, the distribution of switch-to-switch routing path lengths of $k$-hop neighbor storage is shown in Figure \ref{fig:compare_hops_pathlen}, where the optimal values are the shortest path lengths.
Storing 2-hop neighbors significantly reduces the routing path lengths compared with storing 1-hop neighbor.
The average routing path length  of greediest routing with only 1-hop neighbors is 5.749. Including 2-hop neighbors, the value is decreased to  5.199, which is very close to 4.874, the average shortest path length.
However, including 3-hop neighbors does not improve the routing path much compared with using 2-hop neighbors.
Therefore, we decide to store 2-hop neighbors for S2 routing. Although storing 2-hop neighbors requires more state, the number of 2-hop neighbors are bounded by $d^2$, where $d$ is the inter-switch port number, and  this number is much lower than $d^2$ in practice.
As forwarding state is independent of the network size, S2 routing is still highly scalable.

\begin{figure}[tb]
\centering
\includegraphics[width=0.85\linewidth] {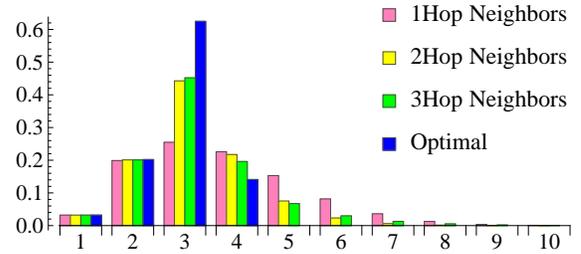}
\vspace{-1ex}
\caption{\small Distribution of routing path lengths using $k$-hop neighbors}
\label{fig:compare_hops_pathlen}
\vspace{-2ex}
\end{figure}

\subsubsection{Impact of the space number}
Proposition \ref{thm:delivery_greediest} holds for any $L \geq 1$.
Therefore, greediest routing can use the coordinates only in the first $d$ spaces, $d<L$, and apply the MCD in the first $d$ spaces ($d$-MCD) as the greedy routing metric.
In an extreme case where $d =1$, greediest routing degenerates to greedy routing on one single ring using the circular distance as the metric.
For $d<L$, the links connecting adjacent switches in the $d, d+1, ..., L$-th spaces are still included in routing decision.
They serve as random links that can reduce routing path length and improve bandwidth.


For all $d$, $1\leq d \leq L$, greedy routing using $d$-MCD provides delivery guarantee and loop-freedom. We evaluate how the value of $d$ affects routing performance by showing the number of spaces $d$ versus the average routing path length of a typical network topology in Figure \ref{fig:PathLengthDiffNumSpace}. The two error bars represent the 10th and 90th percentile values. Only switch-to-switch paths are computed. The optimal results shown in the figure are shortest path lengths, which in average is 2.498.
We find that routing path lengths significantly reduce when the 2nd and 3rd spaces are included in greedy routing. Using more than 4 spaces, the average length is about 2.5 to 2.6, which is close to the optimal value.
Hence greediest routing in  S2 always use as many spaces as switch port capacity allows. Commodity switches have  more than enough ports to support 5 or more spaces.

\begin{figure}[tb]
\centering
\includegraphics[width=0.85\linewidth]{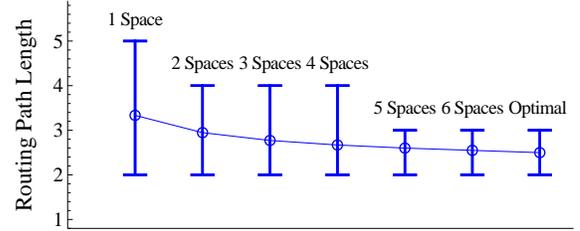}
\caption{\small Routing path length using different numbers of spaces}
\vspace{-2ex}
\label{fig:PathLengthDiffNumSpace}
\vspace{-2.5ex}
\end{figure}

\subsection{Multi-path routing}
Multi-path routing is essential for delivering full bandwidth among servers in a densely connected topology and performing traffic engineering.
 Previous greedy routing protocols can hardly apply existing multi-path algorithms such as equal-cost multi-path (ECMP) \cite{ECMP} and $k$-shortest paths \cite{Jellyfish}, because each switch lacks of global knowledge of the network topology. Consider a potential multi-path method for greedy routing in a single Euclidean space.
 For different flows to a same destination, the source switch intentionally forwards  them to different neighbors by making not-so-greedy decisions. This approach may result longer routing paths.
 In addition, these paths will share a large proportion of overlapped links because all flows are sent to a same direction in the Euclidean space.
 Overlapped links can easily be congested. Therefore, designing  multi-path greedy routing in a single space is challenging.

Greediest routing on S2 supports multi-path routing well due to path diversity across different spaces.
Acoording to Lemma \ref{lemma2}, if a routing protocol reduces the MCD to the destination at every hop, it will eventually find a loop-free path to the destination.
Based on this property, we design a multi-path routing protocol presented as follows.
When a switch $s$ receives the first packet of a new flow whose destination switch $t$ is not $s$, it determines a set $V$ of neighbors, such that for any $v \in V$, $MCD(\vec{X}_v, \vec{X}_t) < MCD(\vec{X}_s, \vec{X}_t)$.
Then $s$ selects one neighbor $v_0$ in $V$ by hashing the 5-tuple of the packet, i.e., source address, destination address, source port, destination port, and protocol type.
 All packets of this flow will be forwarded to $v_0$, as they have a same hash value.
 Hence, packet reordering is avoided.
This mechanism only applies to the first hop of a packet, and on the remain path the packet is still forwarded by greediest routing.
 The main consideration of such design is to restrict path lengths.
 According to our observation from empirical results, multi-pathing at the first hop already provides good path diversity. The pseudocode of the multi-path routing protocol is presented by Algorithm 2 in the appendix.

S2 multi-path routing is also load-aware. As discussed in \cite{Difs2014}, load-aware routing provides better throughput. We assume a switch maintains a counter to estimate the traffic load on each outgoing link.  At the first hop, the sender can select the links that have low traffic load. Such load-aware selection is flow-based: all packets of a flow will be sent to the same outgoing link as the first packet. 

\subsection{Key-based routing}

%
Key-based routing enables direct data access without knowing the IP address of the server that stores the data.
%
%
%
S2 supports efficient key-based routing based on the principle of consistent hashing.
Only small changes are required to the greediest routing protocol.
%

Let $K_a$ be the key of a piece of data $a$.
In S2, $a$ should be stored in $d$ multiple copies at different servers.
%
In S2 key-based routing, a set of globally known hash functions $H_1,H_2,...,H_d$ can be applied to $K_a$. We use $H(K_a)$ to represent a hash value for $K_a$ mapped in $[0, 1]$. The routable address of $K_a$ is defined as $\langle H_1(K_a),H_2(K_a),...,H_d(K_a)\rangle$.
%
For each space $r$, $1\leq r \leq d$, the switch $s$ whose coordinate $x_{s, r}$ is closest
\footnote{Ties should be broken here. One possible approach is to select the switch with larger coordiante.}
to $H_r(K_a)$ among all switches is called the \emph{home switch} of $K_a$ \emph{in space $r$}.
 $K_a$ has at most $d$ home switches in total.
%
%
 A replica of $a$ is assigned to one of the servers connected to the home switch $s$.
In fact, if greediest routing in the first $d$ spaces cannot make progress on switch $s$, then $s$ is a home switch of $K_a$.
S2 supports key-based routing by executing greediest routing to coordiante $\langle H_1(K_a),H_2(K_a),...,H_d(K_a)\rangle$ in the first $d$ spaces. 


\begin{figure}[t]
\centering
\includegraphics[width=0.85\linewidth]{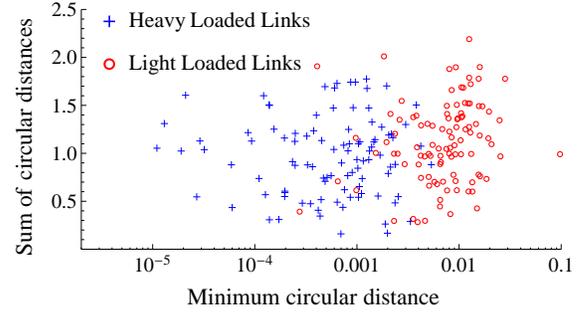}
\vspace{-1ex}
\caption{\small Heavy and light loaded links. $x$-axis: MCD of a link's two endpoints; $y$-axis: sum of CDs of link's two endpoints in all spaces.}
\label{fig:BusyFreeCoor}
\vspace{-2ex}
\end{figure}

\begin{figure}[t]
\centering
\includegraphics[width=0.85\linewidth]{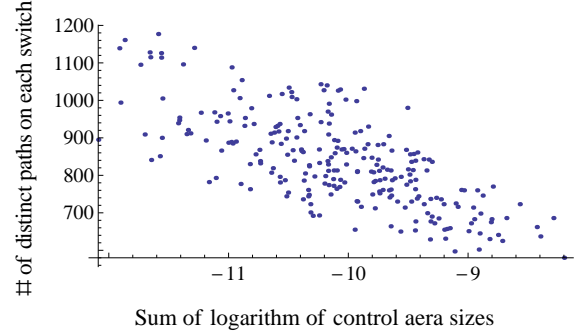}
\vspace{-1ex}
\caption{\small Switch traffic load affected by control area sizes}
\label{fig:ControlAera_paths}
\vspace{-4ex}
\end{figure}

\begin{figure*}[ht]
\centering
\begin{tabular}{p{160pt}p{160pt}p{160pt}}
\includegraphics[width=1\linewidth]{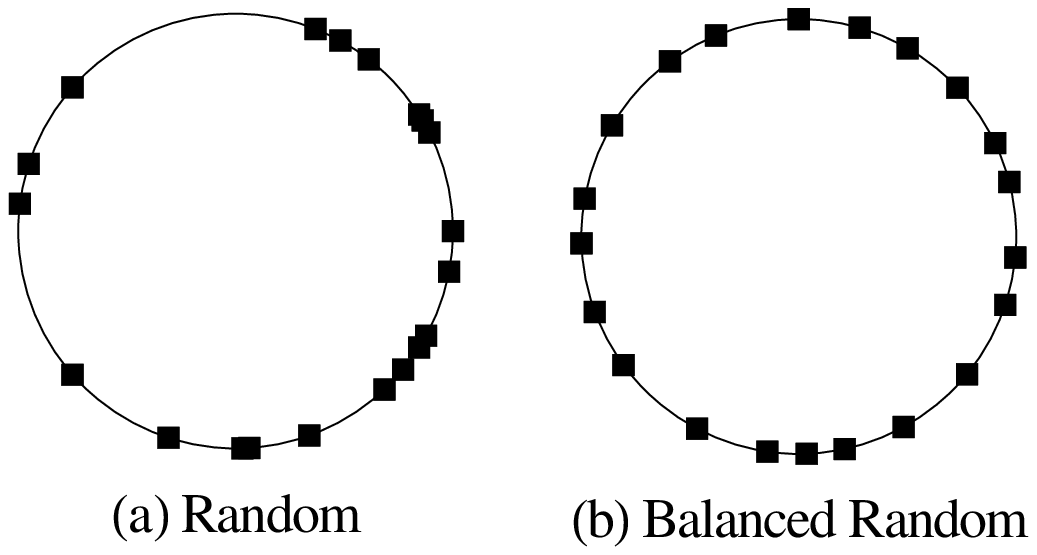}
\vspace{-3ex}
\caption{\small Examples of random and balanced random coordinates}
\label{fig:compareCoordinateGenerate}
&
\includegraphics[width=1\linewidth]{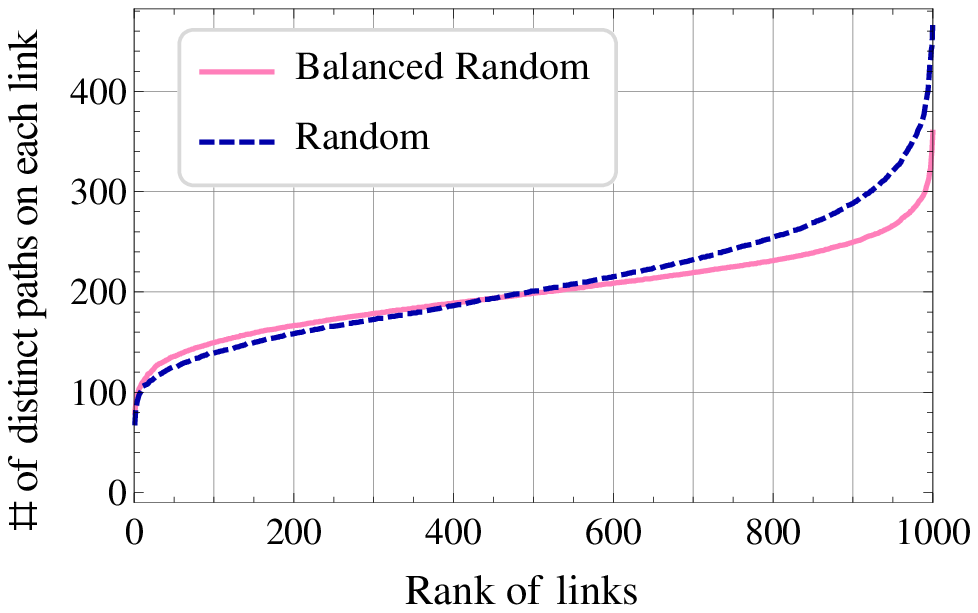}
\vspace{-3ex}
\caption{\small  Distribution of the number of paths on a link}
\label{fig:coordinate_distribute}
&
\includegraphics[width=1\linewidth]{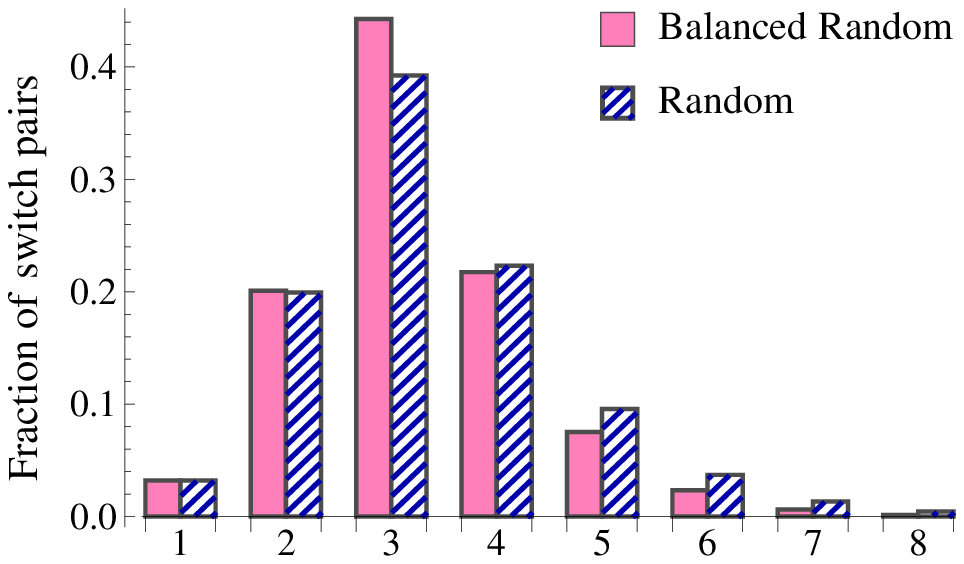}
\vspace{-3ex}
\caption{\small Routing Path Length of two Coordinate Generating Algorithms}
\label{fig:coordinate_distribute_pathlen}
\end{tabular}
\vspace{-5ex}

\end{figure*}

\subsection {Balanced random coordinates}
\label{sec:balance}
Purely uniform random generation of S2 coordinates will probably result in an imbalanced coordinate distribution.
%
%
Figure \ref{fig:compareCoordinateGenerate}(b) shows an example coordinate distribution of 20 switches in a space.
The right half of this ring has much more switches than the left half. Some switches are close to their neighbours while some are not.
Theoretically, among $n$ uniform-randomly generated coordinates, the expected value of the minimum distance between two adjacent coordinates is $\frac{1}{n^2}$, while the expected value of the maximum is $\Theta(\frac{\log n}{n})$ \cite{orderstatics}.
Imbalance of coordinate distribution is harmful to S2 routing in two main aspects. First, greediest routing may intend to choose some links and cause congestion on them. We conjecture as follows. Consider two connected switches $A$ and $B$ whose coordinates are extremely close in one space.  If one of them, say $A$, is the destination of a group of flows, other switches may intend to send the flows  to $B$ if they are unaware of $A$. These flows will then be sent from $B$ to $A$ and congest the link. Second, imbalanced key-value store occurs if switches are not evenly distributed on a ring. Previous work about load balancing in ring spaces cannot be applied here because they do not consider greediest routing.

We perform empirical study of the impact of coordinate distribution to routing loads. In a typical S2 network with $250$ switches and $L=4$, we run greediest routing for all pairs of switches to generate routing paths and then count the number of distinct paths on each link. We find the top 10\% links and bottom 10\% links according to the numbers of distinct paths and denote them by heavy loaded links and light loaded links respectively.
We plot the heavy and light loaded links in a 2D domain as shown in Figure \ref{fig:BusyFreeCoor}, where the $x$-axis is the MCD of a link's two endpoints  and the $y$-axis is the sum of circular distances of a link's two endpoints in all spaces. We find that the frequency of heavy/light loaded links strongly depends on the MCD of two endpoints, but has little relation to the sum of circular distances. If the MCD is shorter, a link is more likely to be heavy loaded. Hence it is desired to avoid two switches that are placed very closely on a ring, trying to enlarge the minimal circular distance for links.

We further study the the impact of coordinate distribution to per-switch loads.
We define the \emph{control area} of switch $s$ in a space as follows:
Suppose switch $s$'s coordinate in this space is $x$, $s$ has two adjacent switches, whose coordinates are $y$ and $z$ respectively.
The control area of $s$ on the ring is the arc between the mid-point of $ \wideparen{y,x} $ and the mid-point of $\wideparen{x,z}$.
The \emph{size of $s$'s control area} in the space is defined as $\frac{1}{2}CD(x,y)+\frac{1}{2}CD(x,z)$.
For the same network as Figure \ref{fig:BusyFreeCoor}, we count the number of different routing paths on each switch.
We then plot this number versus the  sum of logarithm of control area sizes of each switch in Figure \ref{fig:ControlAera_paths}.
It shows that they are negatively related with a correlation coefficient $-0.7179$. Since the sum of control area sizes of all switches is fixed , we should make the control areas as even as possible to maximize the sum-log values. This is also consistent to the load-balancing requirement of key-value storage.
Based on the above observations, we present a \textsc{balanced random coordinate generation} algorithm: 
When a switch $s$ joins the network with $n$ switches, in every space we select two adjacent switches with the maximum circular distance, whose coordinates are $y$ and $z$.
By the pigeonhole principle, $CD(y, z) \geq \frac{1}{n}$.
Then we place $s$ in somewhere between $y$ and $z$.
To avoid being too close to either of $y$ and $z$, we generate $s$'s coordinate $x$ in the space as a random number inside $(y+\frac{1}{3n}, z-\frac{1}{3n})$, so that $CD(x,y) \geq \frac{1}{3n}$ and $CD(x,z) \geq \frac{1}{3n}$.
This algorithm can be used for either incremental or deploy-as-a-whole construction. 
It is guaranteed that the MCD between any pair of switches is no less than $\frac{1}{3n}$. An example of balanced random coordinates is shown in Figure \ref{fig:compareCoordinateGenerate}. The pseudocode is presented by Algorithm 3 in the appendix.

For 10-port 250-switch networks, we calculate the greediest routing path for every pair of switches. We show a typical distribution of routing load (measured by the number of distinct routing paths) on each link in Figure \ref{fig:coordinate_distribute}, where we rank the links in increasing order of load. Compared with purely random coordinates, balanced random coordinates increase the load on under-utilized links (before rank 300) and evidently decrease the load on over-utilized links (after rank 600). About 8\% links of  purely random coordinates have more than 300 paths on each of them, and only 1\%  links of balanced random coordinates have that number. The maximum number of distinct paths that a link is on also decreased from $470 $ to $350$ using balanced random coordinates. Balanced random coordinates provide better fairness among links, and thus improve the network throughput.

Besides link fairness, we also examine the routing path lengths using balanced random coordinates. Fig \ref{fig:coordinate_distribute_pathlen} shows the distribution of switch-to-switch routing path lengths of the same network discussed above. Balanced random coordinates also slightly reduce the routing path lengths.
The average routing path length is decreased from $3.35$ to $3.20$.

\begin{figure*}[ht!]
\centering
\begin{tabular}{p{192pt}p{144pt}p{144pt}}
\centering
\includegraphics[width=0.85\linewidth]{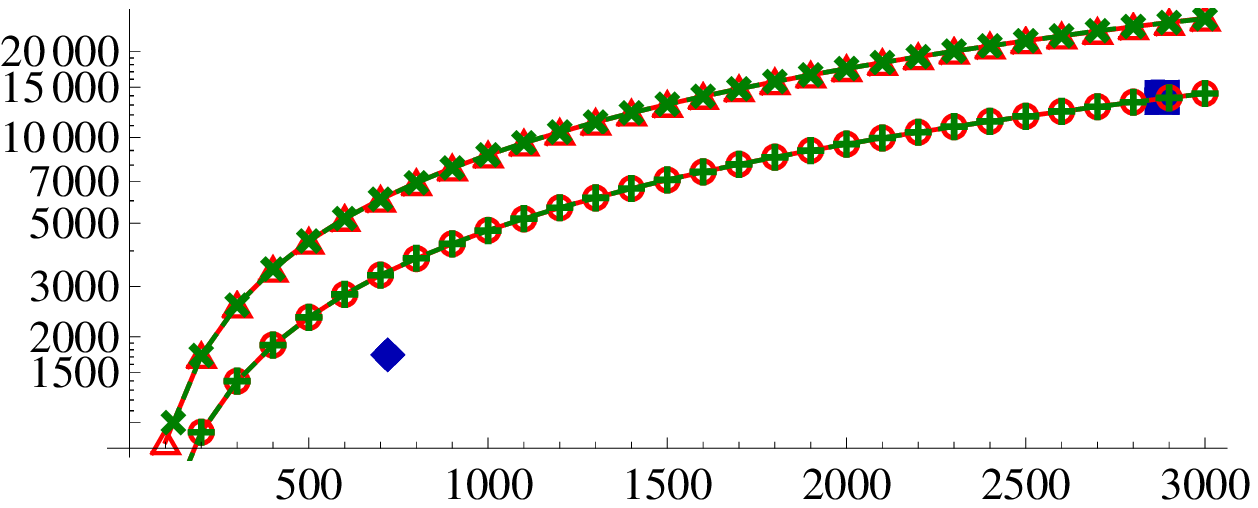}
\includegraphics[width=0.85\linewidth]{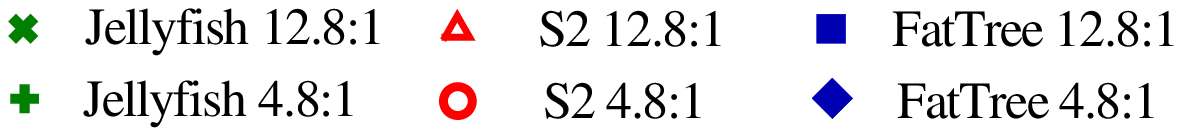}
\vspace{-1ex}
\caption{Bisection bandwidth of S2, FatTree, and Jellyfish. The ratio
(12.8:1 and 4.8:1) is the sever-to-switch ratio.}
\label{fig:bisec}
&
\vspace{-15ex}
\includegraphics[width=1\linewidth] {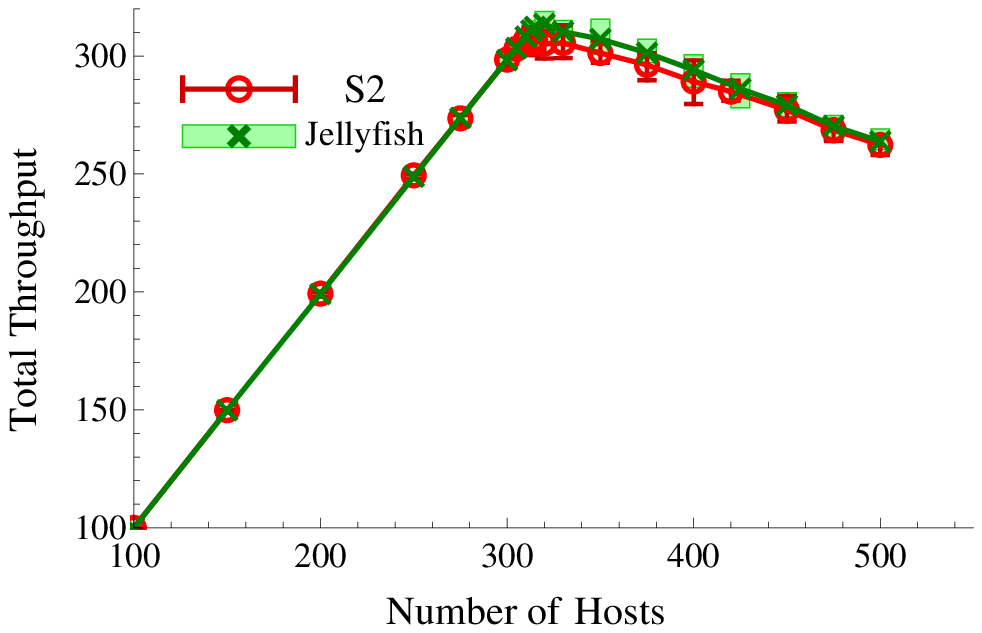}
\vspace{-4ex}
\caption{Ideal throughput of S2 and Jellyfish for a 125-switch network}
\label{fig:idealvsJF}
&
\vspace{-15ex}
\includegraphics[width=1\linewidth] {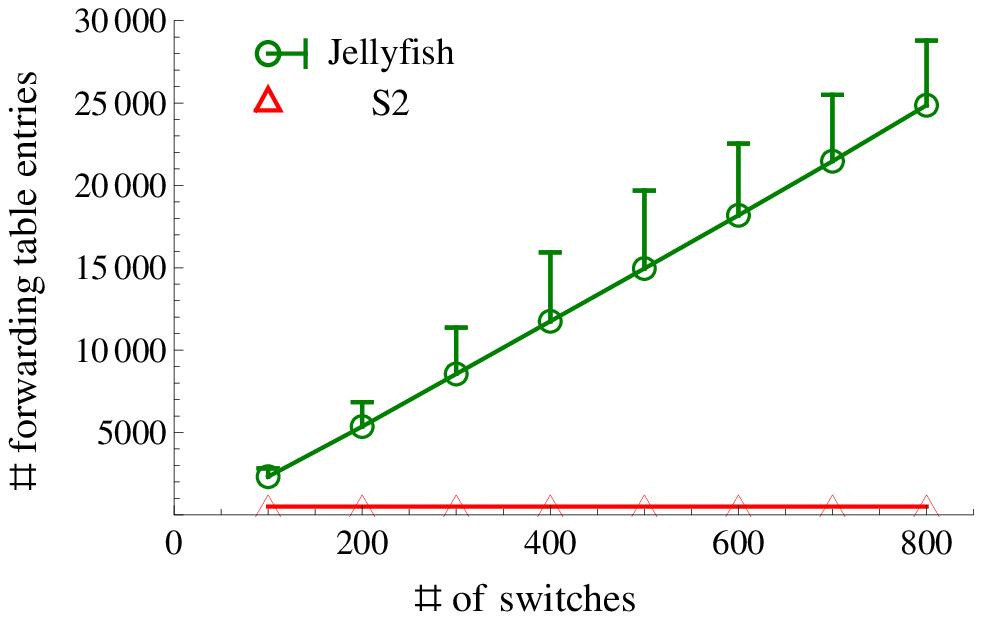}
\vspace{-4ex}
\caption{Forwarding state of S2 and Jellyfish}
\label{fig:RoutingTable}
\end{tabular}
\vspace{-7ex}
\end{figure*}

\section{Evaluation}
\label{sec:evaluation}
In this section, we conduct extensive experiments to evaluate the efficiency, scalability, fairness, and reliability of S2 topologies and routing protocols.
We compare S2 with two recently proposed data center networks, namely Small-World data center (SWDC) \cite{SWDC} and Jellyfish \cite{Jellyfish}.

\subsection{Methodology}
Most existing studies use custom-built simulators to evaluate data center networks at large scale \cite{Hedera} \cite{Scafida} \cite{LEGUP} \cite{SWDC} \cite{REWRITE} \cite{Jellyfish} \cite{Godfrey14}.
We find many of them use a certain level of abstraction for TCP, which may result in inaccurate throughput results.
We develop our own simulator\footnote{We experienced very slow speed when using NS2 for data center networks. We guess the existing studies do not use NS2 due to the same reason. } to perform fine-grained packet-level event-based simulation.
%
TCP New Reno is implemented in detail as the transportation layer protocol.
We simulate all packets in the network including ACKs, which are also routed by greedy routing.
Our switch abstraction maintains finite shared buffers and forwarding tables.

We evaluate the following performance criteria of S2.

\textbf{Bisection bandwidth} describes the network capacity by measuring the bandwidth between two equal-sized part of a network.
%
we calculate the empirical minimum bisection bandwidth by randomly splitting the servers in the network into two partitions and compute the \emph{maximum flow} between the two parts.
%
The minimum bisection bandwidth value of a topology is computed from 50 random partitions. Each value shown in figures is the average of 20 different topologies.

\textbf{Ideal throughput} characterizes a network's raw capacity with perfect load balancing and routing (which do not exist in reality). A flow can be split into infinite subflows which are sent to links without congestion. Routing paths are not specified and flows can take any path between the source and destination. We model it as a  \emph{maximum multi-commodity network flow} problem and solve it using the IBM CPLEX optimizer \cite{CPLEX}. The throughput results are calculated using  a specific type of network traffic, called the \textit{random permutation traffic} used by many other studies \cite{Hedera} \cite{Jellyfish} \cite{Godfrey14}. 
Random permutation traffic model generates very little local traffic and is considered easy to cause network congestion \cite{Hedera}.

\textbf{Practical throughput} is the measured throughput of random permutation traffic routed by proposed routing protocols on the corresponding data center topology. It reflects how a routing protocol can utilize the topology bandwidth. We compare the throughput of S2 with Jellyfish and SWDC for both single-path and multi-path routing.

\textbf{Scalability}. We evaluate forwarding state on switches to characterize the data plane scalability. We measure the number of forwarding entries for shortest-path based routing. However, greedy routing uses distance comparison which does not rely on forwarding entries. Therefore we measure the number of coordinates stored. The entry-to-coordinate comparison actually gives a disadvantage to S2, because storing a coordinate requires much less memory than storing a forwarding entry.

\textbf{Routing path lengths} are important for data center networks, because they have strong impact to both network latency and throughput. For an S2 network, we calculate the routing path length for every pair of source and destination switches and show the average value. 

\textbf{Fairness}. We evaluate throughput and completion time of different flows.

\textbf{Resiliency to network failures} reflects the reliability of the network topology and routing protocol. We evaluate the routing path length and routing success rate under switch failures.

SWDC allows each node to store 2-hop neighbors.
The default SWDC configuration has 6 inter-switch ports.
For SWDC configurations with more than 6 inter-switch ports, we add random links until all ports are used.
For Jellyfish, we use the same implementation of $k$-shortest path algorithm \cite{kshort, kshortCode} as in \cite{Jellyfish}.

Each result shown by a figure in this section, unless otherwise mentioned, is from at least 20 production runs using different topologies.

\begin{figure*}[t!]
\centering
\begin{tabular}{p{160pt}p{160pt}p{160pt}}
\includegraphics[width=1\linewidth]{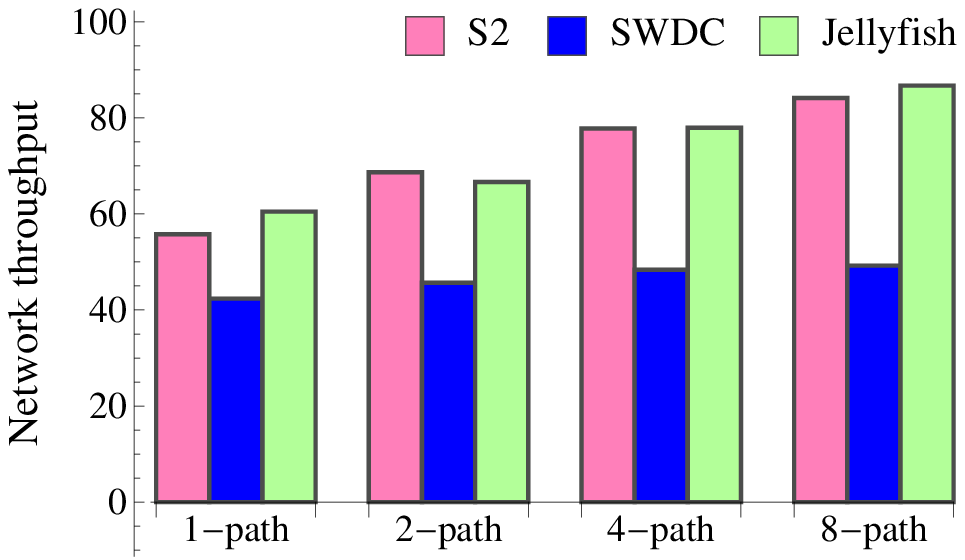}
\caption{\small Throughput of a 250-switch 500-server network}
\label{fig:PkgCompAll}&
\includegraphics[width=1\linewidth]{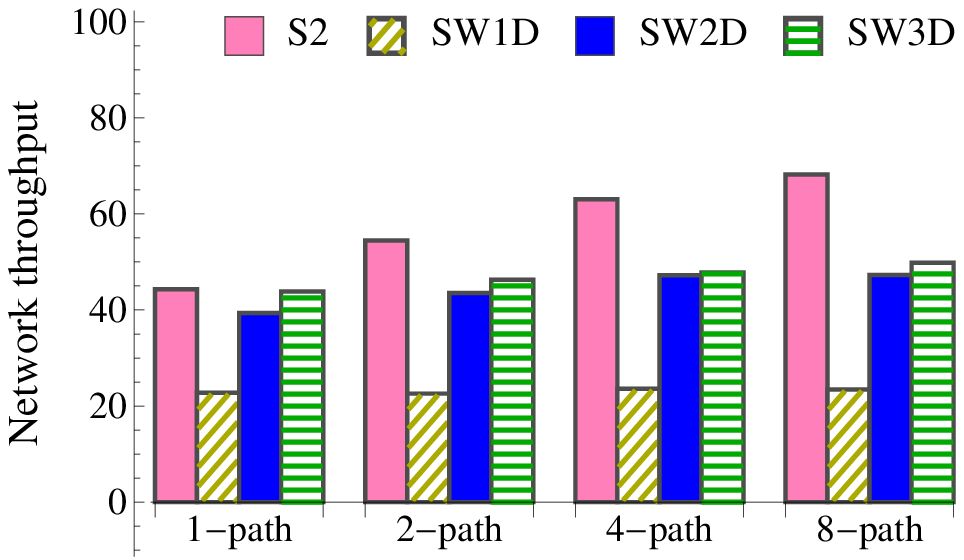}
\caption{\small Throughput of a 400-switch network in SWDC configuration}
\label{fig:THRvsSW400LoadAware}&
\includegraphics[width=1\linewidth]{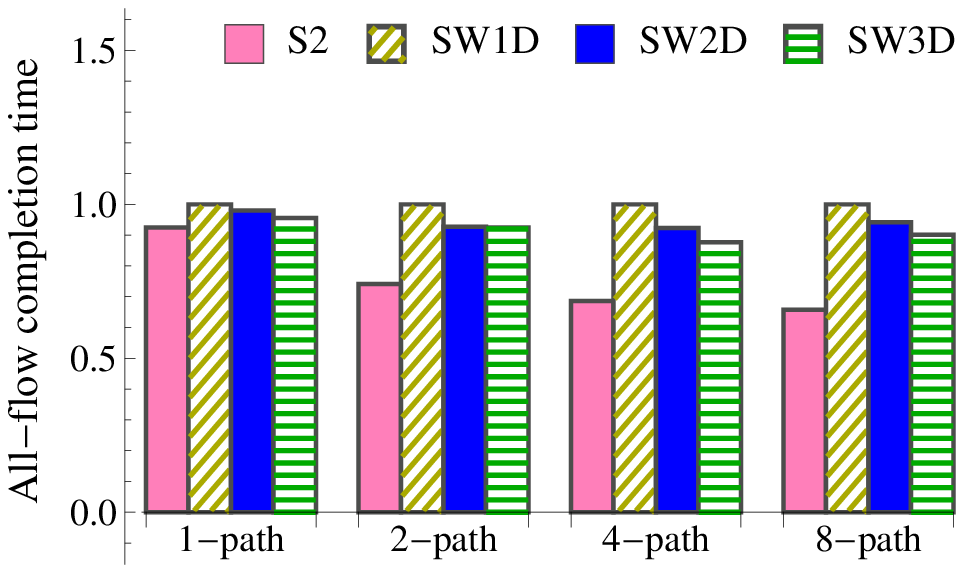}
\caption{\small All-flow completion time}
\label{fig:FCT400LoadAware}
\end{tabular}
\vspace{-7ex}
\end{figure*}

\subsection{Bisection bandwidth}
We compare the minimum bisection bandwidth of S2, Jellyfish, SWDC, and FatTree.
For fair comparison, we use two FatTree networks as benchmarks, a 3456-server 720-switch (24-port) FatTree and a 27648-server 2880 switch (48-port) FatTree.
Note that FatTree can only be built in fixed sizes with specific numbers of ports. The ratio of server number to switch number in above two configurations are 4.8:1 and 12.8:1 respectively.
For experiments of S2 and Jellyfish, we fix the server-to-switch ratio in these two values and vary the number of switches.
In Figure \ref{fig:bisec},
We show the bisection bandwidth of S2, FatTree, and Jellyfish, in the two server-to-switch ratios.
The isolated diamond and square markers represent the minimum bisection bandwidth of FatTree.
%
Both S2 and Jellyfish are free to support arbitrary number of servers and switches. They have identical bisection bandwidth according to our results.
%
%
When using the same number of switches as FatTree (732 and 2880), both S2 and Jellyfish provide substantially higher bisection bandwidth than FatTree.
SWDC only uses a fixed 1:1 server-to-switch ratio and 6-port switches as presented in the SWDC paper \cite{SWDC}.
In such configuration,  S2, SWDC, and Jellyfish have similar bisection bandwidth. However it is not clear whether SWDC can support incremental growth.

\subsection{Ideal throughput}
We model the computation of ideal throughput as a maximum multi-commodity network flow problem: each flow is a commodity without hard demand. We need to find a flow assignment that maximizes network throughput while satisfying capacity constraints on all links and flow conservation. Each flow can be split into an infinite number of subflows and assigned to different paths. We solve it through linear programming using the IBM CPLEX optimizer \cite{CPLEX} and then calculate the maximized network throughput. We show the throughput versus the number of servers of a typical 10-port 125-switch network in Figure \ref{fig:idealvsJF}. When the server number is smaller than 320, the total throughput increases with the server number. After that the network throughput decreases because inter-switch ports are taken to support more servers. S2 is marginally worse than Jellyfish, which has been shown to have clearly higher throughput than FatTree with the same network equipments \cite{Jellyfish}.

\subsection{Scalability}
We consider each coordinate as an entry and compare the number of entries in forwarding tables. In practice, a coordinate requires much less space than a forwarding entry. Even though we give such a disadvantage to S2, S2 still shows huge lead in data plane scalability. Figure \ref{fig:RoutingTable} shows the average and maximum forwarding table sizes of S2 and Jellyfish in networks with 10 inter-switch ports. The number of entries of S2 is no more than 500 and doest not increase when the network grows. The average and maximum forwarding entry numbers of Jellyfish in MPLS implementation \cite{Jellyfish} are much higher. Note the curve of Jellyfish looks like linear but it is in fact $\Theta(N\log N)$. When $N$ is in a relatively small range, the curve of $\Theta(N\log N)$ is close to linear. Using the SWDC configuration, the forwarding state of SWDC 3D is identical to that of S2, and those of SWDC 1D and 2D are smaller.

\begin{figure*}[t!]
\centering
\begin{tabular}{p{160pt}p{160pt}p{160pt}}
\includegraphics[width=1\linewidth]{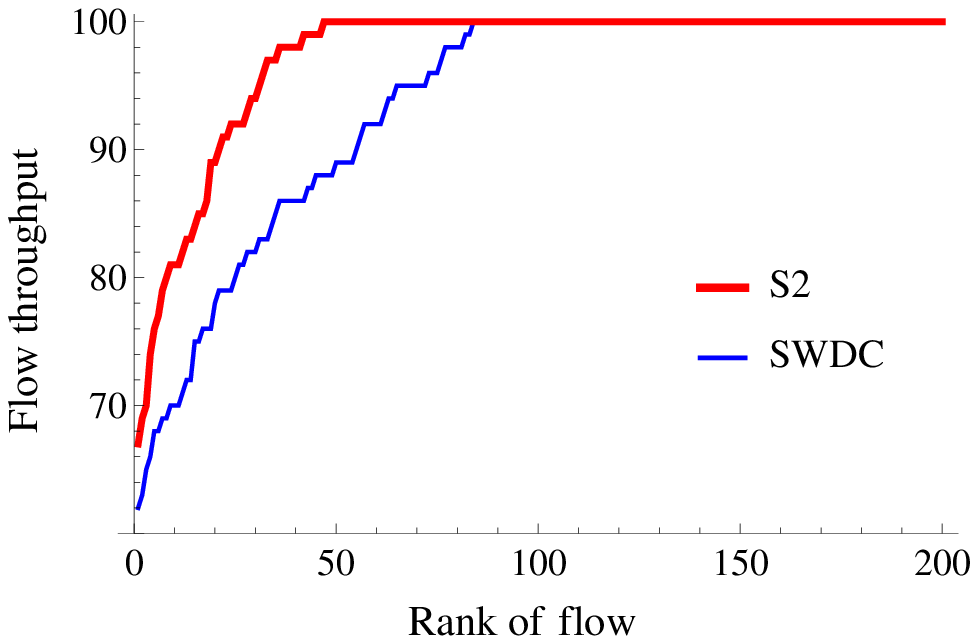}
\vspace{-3ex}
\caption{\small Throughput fairness among flows}
\label{fig:FairnessFlow}&
\includegraphics[width=1\linewidth]{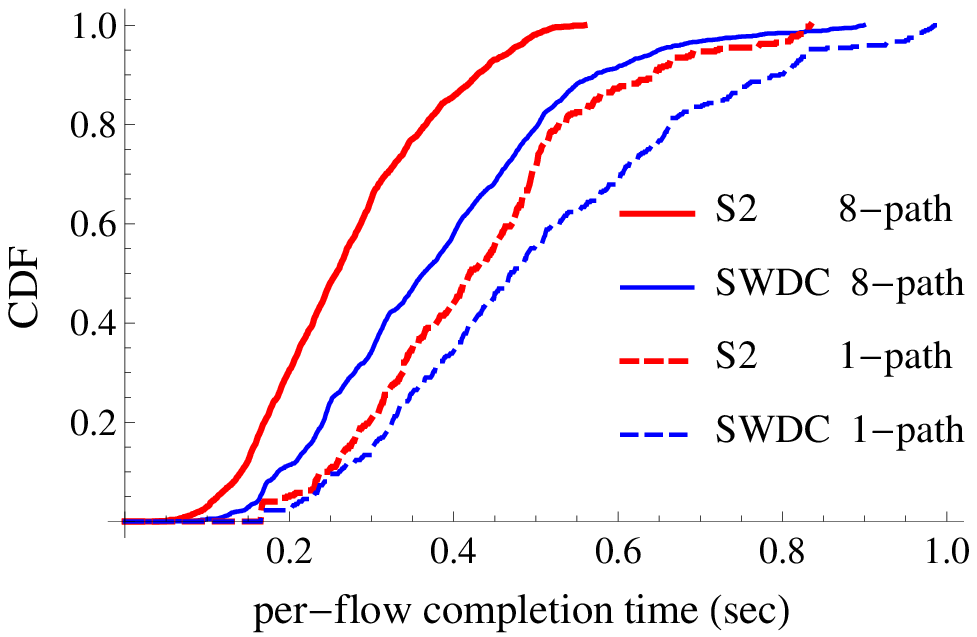}
\vspace{-3ex}
\caption{\small Cumulative distribution of per-flow completion time}
\label{fig:CompletiontimeLoadAware}
&
\includegraphics[width=1\linewidth]{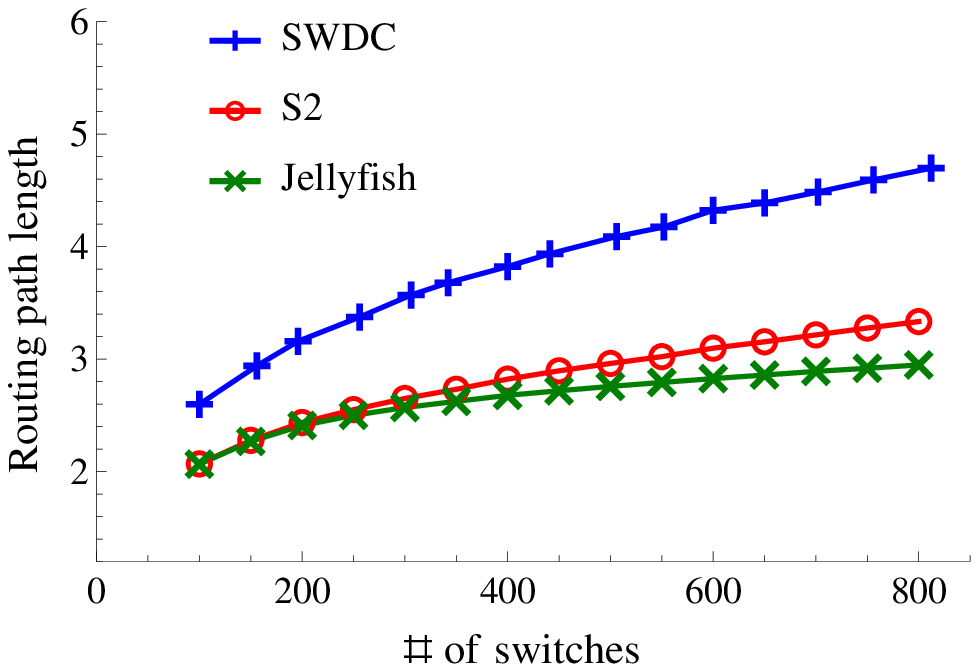}
\vspace{-3ex}
\caption{\small Average routing path length of S2, SWDC, and Jellyfish}
\label{fig:routinglength}
\end{tabular}
\vspace{-8ex}
\end{figure*}

From our experiments on a Dell Minitower with an Intel Core I7-4770 processor and 16GB memory, we also find that it takes hours to compute all pair $8$-shortest paths for Jellyfish with more than 500 switches. Hence it is difficult for switches to compute $k$-shortest paths of a large network in a way similar to link-state routing.


\subsection{Practical throughput}

We conduct experiments to measure the practical throughput of S2, SWDC, and Jellyfish for both single-path and multi-path routing.%
For multi-path routing, the sender splits a flow into $k$ subflows and sends them using S2 multi-path routing.
Packets of the same subflow are forwarded via the same path.
Since the multi-path routing protocol of SWDC is not clearly designed in  \cite{SWDC}, we use a multi-path method similar to that of S2.

In Figure \ref{fig:PkgCompAll} we show the network throughput (normalized to 100) of S2, SWDC, and Jellyfish of a 12-port 250-switch network with 550 servers, using routing with 1, 2, 4, and 8 paths per flow. S2 and Jellyfish have similar network throughput. Using 2-path and 4-path routing, S2 has slightly higher throughput than Jellyfish, while Jellyfish has slightly higher throughput than S2 for 1-path and 8-path.  
Both S2 and Jellyfish overperform SWDC in throughput by about 50\%. We find that multi-path routing improves the throughput of SWDC very little. We conjecture that multi-path greedy routing of SWDC may suffer from shared congestion on some links, since greedy routing paths to a same destination may easily contain shared links in a single space.

In fact, SWDC has three variants (1D Ring, 2D Torus, and 3D Hex Torus) and special configuration (inter-switch port number is 6 and one server per switch). Hence we conduct experiments to compare S2 with all three SWDC networks in the SWDC configuration. Figure \ref{fig:THRvsSW400LoadAware} shows that even under the S2 configuration, S2 provides  higher throughput than all three types of SWDC especially when multi-pathing is used. We only show SWDC 2D in remaining results, as it is a middle course of all three types.

\textbf{Flow completion time:}  We evaluate both all-flow and per-flow completion time of data transmission. Figure \ref{fig:FCT400LoadAware} shows the time to complete transmitting all flows in the same set of experiments as in  Figure \ref{fig:PkgCompAll}. Each flow transmits 16 MB data. S2 takes the least time (0.863 second) to finish all flows. SWDC 2D and 3D also finish all transmissions within 1 second, but use longer time than S2.


\subsection{Fairness among flows}

We demonstrate that S2 provides fairness among flows in the following two aspects.

\textbf{Throughput fairness:} We evaluate the throughput fairness of S2. For the experiments conducted for Figure \ref{fig:PkgCompAll}, we show the distribution of per-flow throughput in Figure \ref{fig:FairnessFlow} where the $x$-axis is the rank of a flow. It shows that S2 provides better fairness than SWDC and more than 75\% of S2 flows can achieve the maximum throughput. Measured by the fairness index proposed by Jain \emph{et al.} \cite{jainindex}, S2 and SWDC 2D have fairness value 0.995741 and 0.989277 respectively, both are very high.

\textbf{Completion time fairness:} We take a representative production run and plot the cumulative distribution of per-flow completion time in Figure \ref{fig:CompletiontimeLoadAware}. We found that S2 using 8-path routing provides both fast completion and fairness among flows -- most flows finish in 0.2 - 0.4 second. 
S2 single-path completes flows more slowly, but still similar to SWDC 8-path routing. Clearly, SWDC single-path has the worst performance in completion time as well as fairness among flows. Jellyfish has similar results as S2, which is not plotted to make the figures clear.

\subsection{Routing Path Length}
Figure \ref{fig:routinglength} shows the average routing path length of S2, SWDC, and Jellyfish by varying the number of switches (12-port). We find that the average path length of S2 is clearly shorter than that of SWDC, and very close to that of Jellyfish, which uses shortest path routing.  For 800-switch networks, the 90th percentile value is 8 for SWDC and 6 for S2 and Jellyfish. The 10th percentile values is 2 for all  S2 and Jellyfish networks, and 3 for all SWDC networks with more than 500 switches. We do not plot the 10th and 90th values in the figure because they make the figure too crowded. Results show that greediest routing in multiple spaces produces much smaller path lengths than greedy routing in a single space.

\subsection{Failure Resiliency}
In this set of experiments, we measure the routing performance of S2, SWDC, and Jellyfish, under switch link failures (a switch failure can be modeled as multiple link failures). We show the routing success rate versus the fraction of failed links in Figure \ref{fig:routingFail}. S2 is very reliable under link failures. When 20\% links fail, the routing success rate is higher than  0.85. SWDC and Jellyfish perform clearly worse than S2. When 20\% links fail, the routing success rate of SWDC is 0.70 and that of Jellyfish is 0.59. S2 uses greedy routing in multiple spaces, hence it is less likely to encounter local minimum under link failure compared to SWDC. Jellyfish has the worst resiliency because it uses pre-computed paths. 


\begin{figure}[t]
\centering
\includegraphics[width=0.75\linewidth]{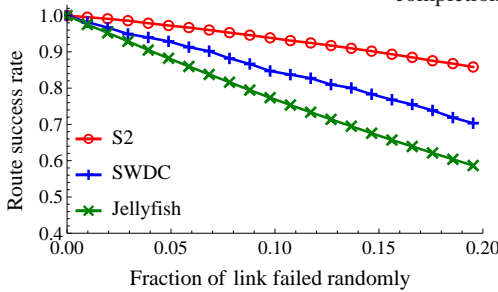}
\vspace{-1ex}
\caption{ Routing success rate versus failure fraction}
\label{fig:routingFail}
\vspace{-5ex}
\end{figure}
\section{Discussion}
\label{sec:discussion}

\subsection{Data center network wiring}
\label{sec:wire}
Labor and wiring expenses consume a significant part of financial budget of building a data center.
S2 can be deployed with cabling optimization to reduce the cost.
In an S2 topology, the majority of cables are inter-switch ones.
Thus we propose to locate the switches physically close to each other so that to reduce cable lengths as well as manual labor.
Compared to FatTree, S2 requires less network switches to obtain a certain bisection bandwidth.
 Therefore the energy consumption, infrastructure and labor cost can be reduced accordingly.


\textbf{Benefits of coordinates:}
It is possible to accommodate the switches of an S2 network inside several standard racks.  These racks can be put close to each other and we suggest to use aggregate cable bundles to connect them.
The coordinates provides a way to reduce inter-rack cables which also helps to arrange the links in order. A virtual space can be divided into several quadrants and we may allocate switches to racks based on corresponding quadrants.
For inner-rack cables, a unique method provided by the nature of coordinates, is using a patch panel that arranges the links in order according to the coordinates. For inter-rack cables, coordinates make it possible to build aggregate bundle wires that are similar to flexible flat cables.

Hamedazimi \emph{et al.} \cite{FireFly2014} proposed to use free-space optical communication in data center networks by putting mirrors and lens on switch racks and the ceiling of data center to reflect beams. Coordinates provide a unique way to locate the switches, and make it able to have these beams neatly ordered.

\subsection{Resiliency to network dynamics}
Shortest path based approaches employ either distributed protocols (e.g., OSPF) or SDN  to accommodate to network dynamics and re-compute shortest paths, which takes time and control traffic to converge.
On the other hand, S2 is more robust to network dynamics as shown in Figure \ref{fig:routingFail} because switches make routing decisions locally and do not need to  re-install forwarding entries.

%
\subsection{Direct server connection}
Although S2 is proposed to interconnect ToR switches, we may also use the S2 topology to connect servers directly and forward packets use S2 routing protocols. Similar approaches are also discussed in CamCube \cite{Symbiotic} and SWDC \cite{SWDC}. There are mainly two key advantages to use this topology. First, greedy routing on a server-centric topology can effectively implement custom routing protocols to satisfy different application-level requirements. This service is called symbiotic routing  \cite{Symbiotic}. Second,  hardware acceleration such as GPUs and NetFPGA can be used for packet switching to improve routing latency and bandwidth \cite{SWDC}. 

\subsection{Switch heterogeneity} \label{sec:heterogeneity}
S2 can be constructed with switches of different port numbers.
The multiple ring topology requires each switch should have at least $2L$ inter-switch ports.
According to Figure \ref{fig:PathLengthDiffNumSpace} and other experimental results, five spaces are enough to provide good network performance.
It is reasonable to assume that every switch in the network has at least 10 inter-switch ports.
Switches with less ports may carry fewer servers to maintain the required inter-switch port number.

\subsection{Possible implementation approaches}
We may use open source hardware and software to implement  S2's routing logic such as NetFPGA. S2's routing logic only includes simple arithmetic computation and numerical comparison and hence can be prototyped in low cost. Besides, S2 can also be implemented by software defined networking such as OpenFlow \cite{OpenFlow}.
According to Devoflow \cite{DevoFlow}, OpenFlow forwarding rules can be extended with local routing decisions, which forward flows that do not require vetting by the controller.  Hence the SDN controller can simply specify the greediest routing algorithm in location actions of switches.
Compared to shortest path routing, S2 has two major advantages to improve the SDN scalability.
First, it reduces the communication cost between switches and the controller. Second there is no need to maintain a central controller that responds to all route queries of the network. Instead, multiple independent controllers can be used for a large network, each of which is responsible to switches in a local area. Such load distribution can  effectively mitigate the scalability problem of a central controller \cite{IMC09} \cite{Benson10}.

\section{Conclusion}
\label{sec:conclusion}

The key technical novelty of this paper is in proposing a novel data center network architecture that achieves all of the three key properties: high-bandwidth, flexibility, and routing scalability. The significance of this paper in terms of impact lies in that greediest routing of S2 is the first greedy routing protocol to enable high-throughput multi-path routing.  We conduct extensive experiments to compare S2 with two recently proposed data center networks, SWDC and Jellyfish. Our results show that S2 achieves the best of both worlds. Compared to SWDC, S2 provides shorter routing paths and higher throughput.  Compared to Jellyfish, S2 demonstrates significant lead in scalability while provides likewise high throughput and bisectional bandwidth. We expect greedy routing using multiple spaces may also be applied to other large-scale network environments due to its scalability and efficiency.


{\small
\bibliographystyle{abbrv}
 \bibliography{bibfile}}

\begin{thebibliography}{10}

\bibitem{hadoop}
Apache hadoop.
\newblock {\em http://hadoop.apache.org/}.

\bibitem{CPLEX}
Ibm cplex optimizer.
\newblock {\em
  http://www-01.ibm.com/software/commerce/optimization/cplex-optimizer/}.

\bibitem{kshortCode}
Implementation of $k$-shortest path algorithm.
\newblock {\em http://code.google.com/p/k-shortest-paths/}.

\bibitem{Symbiotic}
H.~Abu-Libdeh et~al.
\newblock Symbiotic routing in future data centers.
\newblock In {\em Proc. of ACM SIGCOMM}, 2010.

\bibitem{fattree}
M.~Al-Fares, A.~Loukissas, and A.~Vahdat.
\newblock A scalable, commodity data center network architecture.
\newblock In {\em Proc. of ACM SIGCOMM}, 2008.

\bibitem{Hedera}
M.~Al-Fares, S.~Radhakrishnan, B.~Raghavan, N.~Huang, and A.~Vahdat.
\newblock Hedera: dynamic flow scheduling for data center networks.
\newblock In {\em Proceedings of USENIX NSDI}, 2010.

\bibitem{Benson10}
T.~Benson, A.~Akella, and D.~A. Maltz.
\newblock Network traffic characteristics of data centers in the wild.
\newblock In {\em Proceedings of ACM IMC}, 2010.

\bibitem{routingmesh}
E.~Bouillet.
\newblock {\em Path routing in mesh optical networks}.
\newblock John Wiley \& Sons, 2007.

\bibitem{RRGSP}
F.~Chung and L.~Lu.
\newblock The average distance in a random graph with given expected degrees.
\newblock {\em Internet Mathematics}, 2003.

\bibitem{clos}
C.~Clos.
\newblock A study of non-blocking switching networks.
\newblock {\em Bell System Technical Journal}, 1953.

\bibitem{Difs2014}
W.~Cui and C.~Qian.
\newblock Difs: Distributed flow scheduling for adaptive routing in
  hierarchical data center networks.
\newblock In {\em Proc. of ACM/IEEE ANCS}, 2014.

\bibitem{REWRITE}
A.~R. Curtis, T.~Carpenter, M.~Elsheikh, A.~Lopez-Ortiz, and S.~Keshav.
\newblock Rewire: An optimization-based framework for unstructured data center
  network design.
\newblock In {\em Proc. of IEEE Infocom}, 2012.

\bibitem{LEGUP}
A.~R. Curtis, S.~Keshav, and A.~Lopez-Ortiz.
\newblock Legup: using heterogeneity to reduce the cost of data center network
  upgrades.
\newblock In {\em Proc. of ACM CoNEXT}, 2010.

\bibitem{orderstatics}
H.~David and H.~Nagaraja.
\newblock {\em Order Statistics}.
\newblock Wiley, 2004.

\bibitem{mapreduce}
J.~Dean and S.~Ghemawat.
\newblock Mapreduce: simplified data processing on large clusters.
\newblock {\em Communications of the ACM}, 2008.

\bibitem{DCN-survey}
A.~Greenberg, J.~Hamilton, D.~A. Maltz, and P.~Patel.
\newblock The cost of a cloud: Research problems in data center networks.
\newblock {\em ACM Sigcomm CCR}, 2008.

\bibitem{vl2}
A.~Greenberg, J.~R. Hamilton, N.~Jain, S.~Kandula, C.~Kim, P.~Lahiri, D.~A.
  Maltz, P.~Patel, and S.~Sengupta.
\newblock Vl2: a scalable and flexible data center network.
\newblock In {\em Proceedings of ACM SIGCOMM}, 2009.

\bibitem{DCell}
C.~Guo et~al.
\newblock Dcell: a scalable and fault-tolerant network structure for data
  centers.
\newblock In {\em Proc. of ACM SIGCOMM}, 2008.

\bibitem{BCube}
C.~Guo et~al.
\newblock Bcube: a high performance, server-centric network architecture for
  modular data centers.
\newblock In {\em Proc. of ACM SIGCOMM}, 2009.

\bibitem{Scafida}
L.~Gyarmati and T.~A. Trinh.
\newblock Scafida: a scale-free network inspired data center architecture.
\newblock {\em ACM Sigcomm CCR}, 2010.

\bibitem{FireFly2014}
N.~Hamedazimi et~al.
\newblock Firefly: A reconfigurable wireless data center fabric using
  free-space optics.
\newblock In {\em Proc. of ACM SIGCOMM}, 2014.

\bibitem{ECMP}
C.~Hopps.
\newblock Analysis of an equal-cost multi-path algorithm.
\newblock {\em RFC 2992}, 2000.

\bibitem{Dryad}
M.~Isard, M.~Budiu, Y.~Yu, A.~Birrell, and D.~Fetterly.
\newblock Dryad: distributed data-parallel programs from sequential building
  blocks.
\newblock In {\em Proc. of ACM EuroSys}, 2007.

\bibitem{jainindex}
R.~Jain, D.~Chiu, and W.~Hawe.
\newblock A quantitative measure of fairness and discrimination for resource
  allocation in shared computer systems.
\newblock {\em DEC Research Report TR-301}, 1984.

\bibitem{IMC09}
S.~Kandula, S.~Sengupta, A.~Greenberg, P.~Patel, and R.~Chaiken.
\newblock The nature of data center traffic: measurements \& analysis.
\newblock In {\em Proceedings of ACM IMC}, 2009.

\bibitem{SEATTLE}
C.~Kim, M.~Caesar, and J.~Rexford.
\newblock {Floodless in {SEATTLE}: A Scalable Ethernet Architecture for Large
  Enterprises}.
\newblock In {\em Proc. of Sigcomm}, 2008.

\bibitem{RRGP2P}
J.~Kim and R.~Srikant.
\newblock Achieving the optimal steaming capacity and delay using random
  regular digraphs in p2p networks.
\newblock {\em CoRR}, abs/1308.6807, 2013.

\bibitem{MDT-ToN}
S.~S. Lam and C.~Qian.
\newblock {Geographic Routing in $d$-dimensional Spaces with Guaranteed
  Delivery and Low Stretch}.
\newblock In {\em IEEE/ACM Transactions on Networking}.

\bibitem{OpenFlow}
N.~McKeown, T.~Anderson, H.~Balakrishnan, G.~Parulkar, L.~Peterson, J.~Rexford,
  S.~Shenker, and J.~Turner.
\newblock Openflow: Enabling innovation in campus networks.
\newblock {\em SIGCOMM Comput. Commun. Rev.}, 2008.

\bibitem{DevoFlow}
J.~C. Mogul, J.~Tourrilhes, P.~Yalagandula, P.~Sharma, A.~R. Curtis, and
  S.~Banerjee.
\newblock Devoflow: scaling flow management for high-performance networks.
\newblock In {\em Proc. of ACM SIGCOMM}, 2011.

\bibitem{Portland}
R.~N. Mysore et~al.
\newblock Portland: a scalable fault-tolerant layer 2 data center network
  fabric.
\newblock In {\em Proceedings of ACM SIGCOMM}, 2009.

\bibitem{ROME}
C.~Qian and S.~Lam.
\newblock {ROME: Routing On Metropolitan-scale Ethernet }.
\newblock In {\em Proceedings of IEEE ICNP}, 2012.

\bibitem{GDV}
C.~Qian and S.~S. Lam.
\newblock {Greedy Distance Vector Routing}.
\newblock In {\em Proceedings of IEEE ICDCS}, June 2011.

\bibitem{DigitalRealty}
D.~Realty.
\newblock 2013: What is driving the north america/europe data center market?
\newblock {\em
  http://www.digitalrealty.com/us/knowledge-center-us/?cat=Research}.

\bibitem{SWDC}
J.-Y. Shin, B.~Wong, and E.~G. Sirer.
\newblock Small-world datacenters.
\newblock In {\em Proc. of ACM SOCC}, 2011.

\bibitem{Godfrey14}
A.~Singla, P.~B. Godfrey, and A.~Kolla.
\newblock High throughput data center topology design.
\newblock In {\em Proc. of USENIX NSDI}, 2014.

\bibitem{Jellyfish}
A.~Singla, C.-Y. Hong, L.~Popa, and P.~B. Godfrey.
\newblock Jellyfish: Networking data centers randomly.
\newblock In {\em Proc. of USENIX NSDI}, 2012.

\bibitem{PAST}
B.~Stephens, A.~Cox, W.~Felter, C.~Dixon, and J.~Carter.
\newblock {PAST: Scalable Ethernet for Data Centers}.
\newblock In {\em Proceedings of ACM CoNEXT}, 2012.

\bibitem{wang2014similarity}
B.~Wang et~al.
\newblock Similarity network fusion for aggregating data types on a genomic
  scale.
\newblock {\em Nature methods}, 11(3):333--337, 2014.

\bibitem{kshort}
J.~Yen.
\newblock Finding the $k$ shortest loopless paths in a network.
\newblock {\em Manegement Science}, 1971.

\bibitem{BUFFALO}
M.~Yu, A.~Fabrikant, and J.~Rexford.
\newblock Buffalo: Bloom filter forwarding architecture for large
  organizations.
\newblock In {\em Proceedings of ACM CoNEXT}, 2009.

\end{thebibliography}

\begin{appendix}

\textbf{Proof of Lemma \ref{lemma1}.}
\begin{proof} \hspace{20 mm}\\
\textbf{(1)} Let $p$ be the switch closest to $x$ among all switches in the space.\\
\textbf{(2)} The ring of this space is divided by $x_s$ and $x$ into two arcs.%
 At least one this two has length no greater than $\frac{1}{2}$. Let it be $\wideparen{x_s, x}$ with length $L(\wideparen{x_s, x})$. We have $CD(x_s, x) = L(\wideparen{x_s, x}) \leq \frac{1}{2}$.\\
\textbf{(3)} If $p$ is on $\wideparen{x_s, x}$, let the arc between $s$ and $p$ on $\wideparen{x_s, x}$ be $\wideparen{x_s, x_p}$.\\
\textbf{(3.1)} If $s$ has an adjacent switch $q$ whose coordinate is on $\wideparen{x_s, x_p}$, then $L(\wideparen{x_q, x}) < L(\wideparen{x_s, x}) \leq \frac{1}{2}$. Hence $CD(q, x)= L(\wideparen{x_q, x}) < L(\wideparen{x_s, x}) = CD(x_s, x)$.\\
\textbf{(3.2)} If $s$ has no adjacent switch on $\wideparen{x_s, x_p}$, $p$ is $x$'s adjacent switch. Hence $s$ has an adjacent switch $p$  such that $CD(x, x_{p})< CD(x, x_s)$.
\textbf{(4)} If $p$ is not on $\wideparen{x_s, x}$, we have an arc $\wideparen{x_s, x, x_p}$. For the arc $\wideparen{x, x_p}$ on $\wideparen{x_s, x, x_p}$, we have $L(\wideparen{x, x_p}) < L(\wideparen{x_s, x}) $. (Assume to the contrary if $L(\wideparen{x, x_p}) \geq L(\wideparen{x_s, x})$. Then we cannot have $CD(x, x_{p})< CD(x, x_s)$. There is contradiction.)\\
\textbf{(4.1)} If $s$ has an adjacent switch $q$ whose coordinate is on $\wideparen{x_s, x, x_p}$, then $L(\wideparen{x_q, x}) < L(\wideparen{x_s, x}) \leq \frac{1}{2}$. Hence $CD(q, x)= L(\wideparen{x_q, x}) < L(\wideparen{x_s, x}) = CD(x_s, x)$.\\
\textbf{(4.2)} If $s$ has no adjacent switch on $\wideparen{x_s, x, x_p}$, $p$ is $x$'s adjacent switch. Hence $s$ has an adjacent switch $p$  such that $CD(x, x_{p})< CD(x, x_s)$.\\
\textbf{(5)} Combining (3) and (4), $s$ always has an adjacent switch $s'$ such that $CD(x, x_{s'})< CD(x, x_s)$.\\
\end{proof}

\textbf{Proof of Lemma \ref{lemma2}.}
\begin{proof} \hspace{20 mm}\\
\textbf{(1)} Suppose the minimum circular distance between $s$ and $t$ is defined by their circular distance in the $j$th space, i.e. $CD(x_{tj}, x_{sj})= MCD_L(\vec{X}_s, \vec{X}_{t})$.\\
\textbf{(2)} In  the $j$th space, $t$ is the switch with the shortest circular distance to $x_{tj}$, which is $CD(x_{tj},x_{tj})=0$. Since $s$ is not $t$,  $s$ is not the switch with the shortest circular distance to $x_{tj}$ ,because any two coordinates are different.\\
\textbf{(3)} Based on Lemma \ref{lemma1}, $s$ has an adjacent switch $s'$ such that \\ $CD(x_{tj}, x_{s'j})< CD(x_{tj}, x_{sj})$.\\
\textbf{(4)} $MCD_L(\vec{X}_{s'}, \vec{X}_{t}) \leq CD(x_{tj}, x_{s'j}) < CD(x_{tj}, x_{sj}) =  MCD_L(\vec{X}_s, \vec{X}_{t})$.\\
\textbf{(5)} Since $v$ is the switch that has the shortest MCD to $\vec{X}_{t}$ among all neighbors of $s$, we have \\
$MCD_L(\vec{X}_{v}, \vec{X}_{t}) \leq MCD_L(\vec{X}_{s'}, \vec{X}_{t})< MCD_L(\vec{X}_s, \vec{X}_{t})$.\\
\end{proof}

\textbf{Proof of Proposition \ref{thm:delivery_greediest}.}
\begin{proof} \hspace{20 mm}\\
\textbf{(1)}Suppose switch $s$ receives a packet whose destination switch is $t$. If $s = t$, the destination host is one of the servers connected to $s$. The packet can be delivered. \\
\textbf{(2)} If $s \neq t$, according to Lemma \ref{lemma2}, $s$ will find a neighbor $v$ such that $MCD_L(\vec{X}_v, \vec{X}_{t}) < MCD_L(\vec{X}_s, \vec{X}_{t})$, and forward the packet to $v$. \\
\textbf{(3)} The MCD from the current switch to the destination coordinates strictly reduces at each hop. Greediest routing keeps making progress. Therefore, there is no routing loop. Since the number of switches is finite, the packet will be delivered to $t$.\\
\end{proof}

\begin{codebox}
\Procname{\proc{ \textbf{Algorithm 1.} Greediest routing on switch $s$}}
\textbf{input:} Coordinates of all neighbors,\\ \ \ \ \ \ \ \ \  destination addresses $\langle \vec{X}_{t}, ID \rangle$.  \\
\li \If $\vec{X}_{s} = \vec{X}_{t}$
\li \Then $h \gets$ the server connected to $s$, with identifier $ID$
\li    Forward the packet to $h$
\li  \Return; \End
\li Compute 
 $MCD_L(\vec{X}_v,\vec{X}_t)$ for all $s$'s neighbor switch $v$
\li Find $v_0$ such that $MCD_L(\vec{X}_{v0},\vec{X}_t)$ is the smallest
\li Forward the packet to $v_0$
\End
\end{codebox}

\begin{codebox}
\Procname{\proc{\textbf{Algorithm 2.} Multi-path routing on switch $s$}}
\textbf{input:} Coordinates of all neighbors,,\\ \ \ \ \ \ \ \ \ destination addresses $\langle \vec{X}_{t}, ID \rangle$  \\
\li \If $\vec{X}_{s} = \vec{X}_{t}$
\li \Then $h \gets$ the server connected to $s$, with identifier $ID$;
\li    Forward the packet to $h$;
\li  \Return; \End
\li \If the packet is not from a server connected to $s$
\li \Then Perform greediest routing;
\li \Return \End
\li $V \gets \emptyset$;
\li \For each neighbor $v$ of $s$
\li \ \ \ \ \ \ \ \textbf{if} $MCD_L(\vec{X}_v,\vec{X}_t) < MCD_L(\vec{X}_s, \vec{X}_t)$ \textbf{then} $V \gets V \cup \{v\} $
\li Select $v_0$ from $V$ by hashing the source and destination\\ addresses and ports;
\li Forward the packet to $v_0$. \End
\end{codebox}

\begin{codebox}
\Procname{$\proc{\textbf{Algorithm 3.} Balanced random coordinate generation}$}
\textbf{input:} Current $n$ coordiantes $x_1, x_2, ..., x_n$ in a circular space  \\
\textbf{output:} One new coordiante $x_{new}$ \\
\li \textbf{if} $n=0$   \textbf{then} \Return $RandomNumber(0,1)$
\li \If $n=1$
\li \Then $a \gets x_1, b \gets x_1+1$ \End  
\li \Else find $x_{r1},x_{r2}$ among $x_1,x_2, ..., x_n$ such that \\
\ \ \ \ \ \ \ \ \  $x_{r1}<x_{r2}$ and $CD(x_{r1},x_{r2})$ is the smallest.
\li \If $x_{r2}-x_{r1} < \frac{1}{2} $
\li \Then $a \gets x_{r1}, b \gets x_{r2} $ \End
\li \Else $a \gets x_{r2}, b \gets x_{r1}+1 $\End \End
\li $t \gets RandomNumber(a+\frac{1}{3n},b-\frac{1}{3n})$
\li \textbf{if} $t>1$ \textbf{then} $t \gets t - 1 $
\li \Return $t$
\end{codebox}

\end{appendix}


\end{document}